\documentclass[twocolumn,journal,final]{IEEEtran}
\usepackage{algorithmicx,fullpage,epstopdf,tikz}
\usepackage{algorithm}
\usepackage{subfigure}
\usepackage{amsmath,epsfig,paralist,verbatim,tikz,array,hyperref}
\usepackage[european resistor, european voltage, european current]{circuitikz}
\usetikzlibrary{arrows}
\usepackage{color}
\usepackage{amsfonts}
\usepackage{amssymb}
\usepackage{epsfig,psfrag}
\usetikzlibrary{arrows,shapes,positioning,snakes}
\usetikzlibrary{decorations.markings,decorations.pathmorphing,decorations.pathreplacing}
\usetikzlibrary{calc,patterns,shapes.geometric,topaths}
\usepackage{epstopdf}
\usepackage{pgfplots}
\newtheorem{definition}{Definition}[section]
\newtheorem{theorem}{Theorem}[section]

    \newtheorem{lemma}{Lemma}[section]

\newcommand{\Y}{\mathbb{Y}} \newcommand{\X}{\mathbb{X}}

\newcommand{\poll}{\Req_{L}}

\newcommand{\tp}{P}
\newcommand{\G}{\mathcal{G}}
\newcommand{\belief}{{\pi}}
\newcommand{\tbelief}{\pi^0}
\newcommand{\cbelief}{l^0}
\newcommand{\logoprob}{o}

\newcommand{\noise}{\epsilon}

\newcolumntype{L}{>{\centering\arraybackslash}m{1.92cm}}
\newcommand{\obseq}{\mathcal{Y}}
\newcommand{\omat}{\mathcal{O}}
\newcommand{\Req}{\mathcal{R}}

\newcommand{\additional}{\mathcal{E}}

\newcommand{\priv}{\eta}

\newcommand{\A}{\mathcal{A}}

\newcommand{\beq}{\begin{equation}}
\newcommand{\eeq}{\end{equation}}
\newcommand{\E}{\mathbf{E}}
\newcommand{\ones}{\mathbf{1}}

\newcommand{\loprob}{\bar{R}}
\newcommand{\history}{\mathcal{H}}
\newcommand{\full}{\mathcal{F}}

\newcommand{\argmax}{\operatorname{argmax}}
\newcommand{\argmin}{\operatorname{argmin}}

\newcommand{\Bs}{R^\pi} 
\newcommand{\ta}{\tilde{a}}
\newcommand{\sigs}{\sigma}
  \def\1{{\mathbf 1}}
\newcommand{\reals} {\Bbb{R}}

\newcommand{\qed} {{$\hfill\blacksquare$}}

\newcommand{\lr}{\leq_r}

\newcommand{\lbelief}{l}

    \newcommand{\Ts}{T}
    \newcommand{\ca}{c_a}
\newcommand{\p}{\prime}
    
    \newcommand{\probe}{\belief}
\newcommand{\response}{a}

\newcommand{\utility}{u}
\newcommand{\budget}{I}
\newcommand{\dataset}{\mathcal{D}}
    
\newcommand{\tindx}{t}
\newcommand{\Tindxter}{T}

\begin{document}
\title{Online Reputation and Polling Systems:  Data Incest, Social Learning and Revealed Preferences}

\author{Vikram~Krishnamurthy,~\IEEEmembership{Fellow,~IEEE}  William Hoiles,~\IEEEmembership{Student Member,~IEEE} %
\thanks{V. Krishnamurthy (e-mail: vikramk@ece.ubc.ca)  and W. Hoiles (email: whoiles@ece.ubc.ca) are
  with the Department of Electrical and Computer
Engineering, University of British Columbia, Vancouver, V6T 1Z4, Canada. 

Parts of Sec.I and II of this paper appear in our tutorial paper \cite{KP14}.
}

}

\maketitle

\begin{abstract}

This paper considers 
  online reputation and polling systems where individuals make recommendations based on their private observations and recommendations of friends.
 Such  interaction of individuals and their social  influence  is modelled as social learning on a directed acyclic graph.
 Data incest (misinformation propagation)   occurs due to  unintentional re-use of identical actions in the formation of public belief in social learning; the information gathered by each agent is mistakenly considered to be independent. This results in overconfidence and bias in estimates of the state.
   Necessary and sufficient conditions are given on the structure of information exchange graph  to mitigate data incest. Incest removal algorithms are presented.
Experimental results on human subjects are presented to illustrate the effect of social influence  and data incest on decision making. 
These experimental results indicate that  social learning protocols require careful design to handle and mitigate data incest.
The incest removal algorithms are illustrated in an expectation polling system where participants in a poll respond with a summary of their friends' beliefs. Finally, the principle of revealed preferences arising in micro-economics theory is used to parse Twitter datasets to determine if social sensors are utility maximizers and then determine their utility functions.  

 \end{abstract}


\section{Introduction}
\label{sec:introduction}

Online reputation systems (Yelp, Tripadvisor, etc.) are of increasing importance in measuring social opinion.
They can be viewed as sensors of social opinion - they
 go beyond physical sensors  since user opinions/ratings (such as human evaluation of a restaurant or movie)   are impossible to measure via  physical sensors.
 Devising a fair online reputation system involves constructing a  data  fusion system that combines  estimates of individuals  to generate an unbiased estimate.
 This 
 presents unique challenges from a statistical signal processing and data fusion point of view.
First,  humans  interact with and influence other humans since  ratings posted on online reputation systems strongly influence the behavior of  individuals.\footnote{81\% of hotel managers  regularly check Tripadvisor reviews  \cite{IMS11}.   A one-star increase in the Yelp rating maps to 5-9 \% revenue increase~\cite{Luc11}. \label{foot}}
This  can result in  non-standard information patterns due to correlations introduced by the structure of the underlying social network.
Second, due to privacy
concerns,
humans rarely reveal  raw observations of the underlying state of nature.
Instead, they reveal their decisions 
(ratings, recommendations, votes) which can be viewed as a low resolution (quantized) function of their raw measurements and interactions with other individuals.

\subsection*{Motivation}




This paper models how {\em data incest} propagates amongst individuals in online reputation and polling systems.
%
Consider  the following example comprising a
multi-agent system  where  agents seek to estimate an underlying state of nature.
An  agent visits a restaurant and obtains  a noisy private measurement of the state  (quality of food).
She then rates the restaurant as  excellent  on an online reputation website.
   Another agent is influenced by this rating, visits the restaurant, and also gives  a good rating on the online reputation website.
  The first agent visits the  reputation  site  and notices that another agent has also given the restaurant a good rating - this double confirms her rating and she enters another good rating. 
     In a fair 
  reputation
  system,  such ``double counting" or ``data incest'' should have been prevented by making the first agent  aware that the rating of the second agent was influenced by her own rating.
  The information exchange between the agents is   represented by the directed graph
  of Fig.\ref{fig:samplem}. The fact
 that there are two distinct paths (denoted in red) between Agent 1 at time 1 and Agent 1 at time 3  implies that 
 the information of Agent 1 at time 1 is double counted thereby leading to a data incest event.
Such  
data incest results in a  bias in the estimate of the underlying state.


\begin{figure}[h]
\centering
{\includegraphics[scale=0.2]{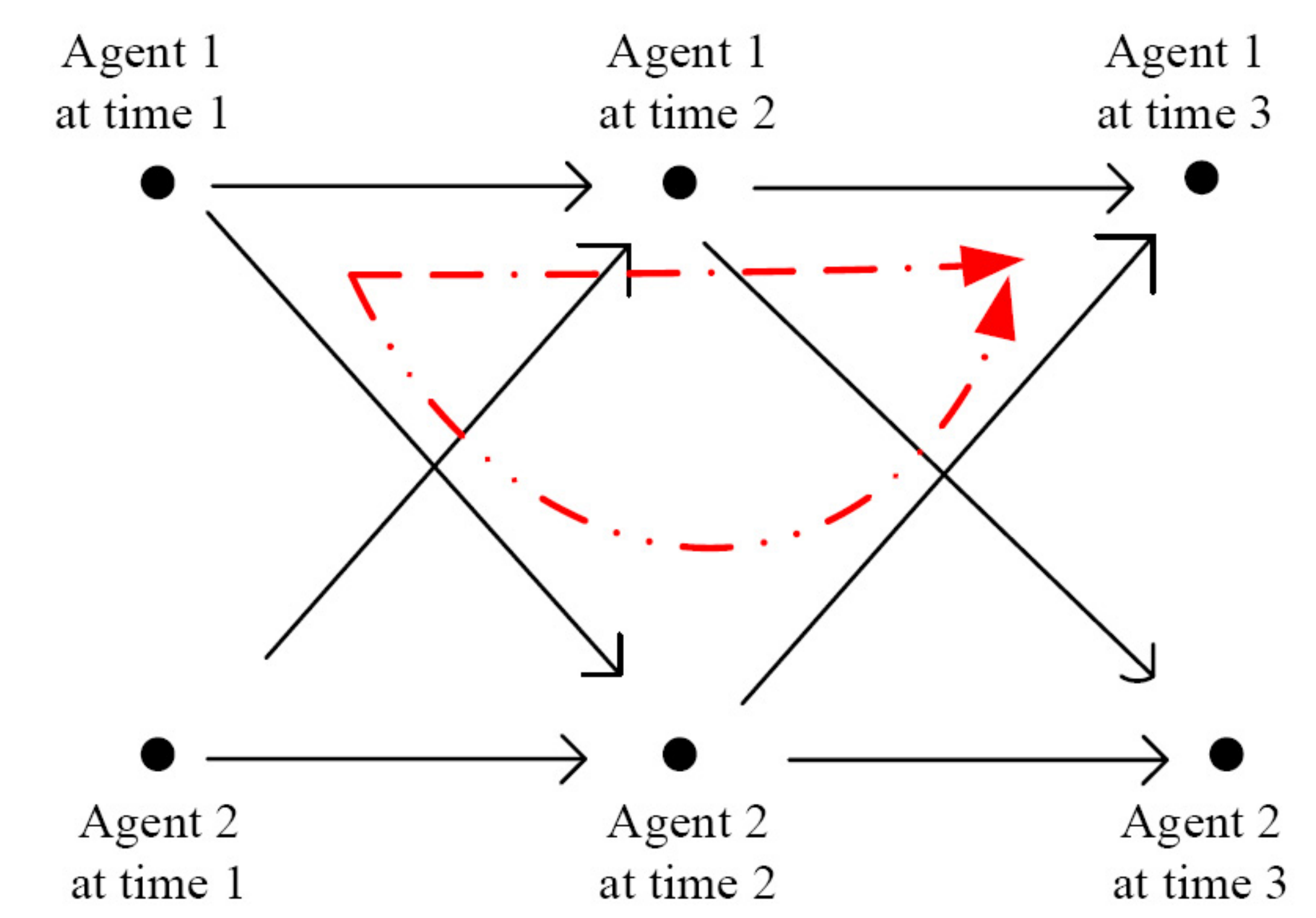}}
\caption{ Example of the information flow  in a social network with two agents  and over three event epochs. The arrows represent exchange of information.}
\label{fig:samplem}
\end{figure}


\subsection*{Main Results and Organization}
This paper has four parts. The first part, namely
Sec.\ref{sec:classicalsocial}, deals with the design of a fair online reputation system.  Sec.\ref{sec:classicalsocial} formulates the data incest problem in a multi agent system where individual agents perform social learning
and  exchange information on  a sequence of directed acyclic graphs. The aim is to develop a  distributed data
fusion protocol which incorporates social influence constraints and provides an  unbiased estimate of the state of nature at each node. Protocol \ref{protocol:socialcons} in Sec.\ref{sec:repute} gives the complete design template of how incest can be
avoided in the online reputation system.  
It is shown that by choosing the costs to satisfy reasonable conditions, the recommendations made by individuals are {\em ordinal} functions of their private observations and {\em monotone}
in the prior information. This means that the Bayesian social learning follows simple intuitive rules and is therefore, a useful idealization.
Necessary and sufficient conditions for exact incest removal subject to a social influence constraint are given.

The second part of the paper, namely 
Sec.\ref{sec:expt}, analyzes the data of an  actual experiment that we performed on human subjects to determine how social influence affects decision making. In particular,
information flow patterns from the experimental data indicative of social learning and data incest are described. The experimental results illustrate the effect of social influence.

The third part of the paper, namely Sec.\ref{sec:vote}, describes how the data incest problem formulation and incest removal algorithms 
 can be applied to an expectation polling system.
Polls seek to estimate the fraction of a population that support a political party,  executive decision, etc.
In {\em intent} polling,
individuals are sampled and asked who they intend to vote for. In {\em expectation} polling \cite{RW10}  individuals are sampled and  asked who they believe will win the election.
It is intuitive  that expectation polling is more accurate than intent polling; since
 in expectation polling an individual  considers its own intent together with the intents of its friends.\footnote{\cite{RW10} analyzes all US presidential electoral college results from 1952-2008 where both intention and expectation polling were conducted and shows a remarkable
result: In 77 cases where expectation and intent polling pointed to different winners, expectation polling was accurate 78\% of the time! The dataset from the American National
Election Studies comprised of voter responses to two questions:\\
{\em Intention}: Who will you vote for in the election for President?\\
{\em Expectation}: Who do you think will be elected President in November?}
If the average degree of nodes in the network is $d$, then  the savings in the number of samples is by a factor of $d$, since a
randomly chosen node summarizes the results form $d$ of its friends. However, the variance and bias of the estimate
depend on the social network structure, and data incest can strongly bias the estimate. We illustrate how the posterior distribution of the leading candidate in the poll can be estimated based on incestious estimates.

Social learning assumes that agents (social sensors) choose their actions by maximizing a  utility function.
The final part of the paper, namely Sec.\ref{sec:revealed},  uses the principle of revealed preferences  as a constructive test to determine: Are social sensors utility optimizers in their response to external influence? We present a remarkable result arising in microeconomics theory called Afriat's theorem
 \cite{Afr67,Var82}  which provides a necessary and sufficient condition for a finite dataset $\dataset$ to have originated from a utility maximizer. 
 The result is illustrated  by tracking real-time tweets and retweets for reputation based review Twitter accounts.
 In particular, Sec.\ref{sec:revealed} shows that social learning associated with reputation agencies has the 
 structure illustrated in Fig.\ref{fig:IntroConnect}. That is,  the sentiment of tweets and number of followers (reputation) constitute the publicly available
 information that drives the social learning process.
 If the sentiment of the tweet published by the reputation agency improves its reputation, 
  then this reinforces the reputation agency's belief that it adequately reviewed the content. This is an example of data incest. 
 

 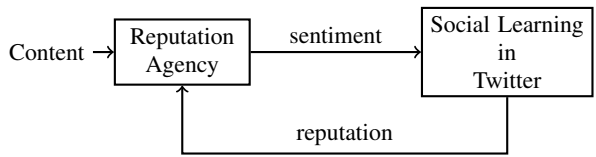
\begin{figure}[h]
  \centering
\begin{tikzpicture}[font = \normalsize, scale =0.9,transform shape, american voltages]

\node (CN) at (-4.0,0) {Content};
\node[draw,rectangle,fill=white,thick,text width=5em,align=center] (RA) at (-2.0,0) {Reputation Agency}; 
\node[draw,rectangle,fill=white,thick,text width=6.5em,align=center] (UM) at (2.8,0) {Social Learning \\ in \\ Twitter}; 

\draw[thick, black, ->] (CN.east) --  (RA.west);
\draw[thick, black, ->] (RA.east) -- node[midway, above] {sentiment} (UM.west);
\draw[thick, black, ->] (UM.south) -- (2.8,-1.5) -| node[near start, above] {reputation} (RA.south) ;

\end{tikzpicture}
  \caption{Schematic of dynamics of the reputation agency and the Twitter network. Content (i.e. games, movies, books) is provided to the reputation agency which then publishes its review as a  tweet. The Twitter network then responds to the sentiment of this tweet by retweeting.
 This may  increase or decrease the  reputation of the agency.} 
\label{fig:IntroConnect}
\end{figure}

\subsection*{Related work}  
\cite{Gra78} is a seminal paper in collective human behavior.
 The book \cite{Cha04} contains a complete treatment of social learning models.
Social learning has been used widely in economics, marketing, political science and sociology  to model the behavior of financial markets, 
crowds, social groups and social networks; see \cite{BHW92,AO11,KT12,Cha04} and numerous references therein. Related models have been studied in  sequential decision making  \cite{CH70}  and statistical signal processing \cite{CSL13,KP14}.
A tutorial exposition of social learning in and sensing is given in our recent paper \cite{KP14}.
Online reputation systems are reviewed and studied in \cite{Del03,JIB07,LNJ10}.
Information caused by influential agents 
 is investigated in 
\cite{AO11,KH13}.   In \cite{XSK11} examples are given that show that if
 just 10\% of the population holds and unshakable belief, their belief will be adopted by the majority of society.

The role of social influence in decision making (which we consider in Sec.\ref{sec:incest}) is studied in \cite{BFJ12}.
The expectation polling \cite{RW10} (which we consider in Sec.\ref{sec:vote}) is a  form of   {\em social sampling}  \cite{DKS12} where
  participants in a poll respond with a summary of their friends
responses.  \cite{DKS12} analyzes the effect of the social  network structure on the bias and variance of expectation polls.
Social sampling has interesting parallels with  the so called Keynesian beauty contest, see for example \url{https://en.wikipedia.org/wiki/Keynesian_beauty_contest}
for a discussion.

Data incest arises in other areas of electrical engineering.
 The  so called the count to infinity problem in the distance vector routing protocol in packet switched networks \cite{RT99}  is a type of misinformation propagation.
 Data incest also arises in  Belief Propagation (BP) algorithms in computer vision and error-correcting coding theory. 
BP algorithms require passing local messages over the graph (Bayesian network) at each iteration. 
For graphical models with loops, BP algorithms are only approximate due to the over-counting of local messages \cite{YFW05} which is similar to data incest.  The algorithms presented in this paper can remove  data incest in Bayesian social learning over non-tree  graphs that satisfy a topological constraint.
In \cite{KH13}, data incest is considered in a network where agents exchange their private belief states (without social learning).  Simpler versions 
were studied in~\cite{Aum76,GP82}.  

Regarding revealed preferences,  highly influential works in the economics
literature include those by 
Afriat \cite{Afr67,Afr87} and
Varian  (chief economist at  Google). In particular Varian's work  includes
measuring the welfare effect of price discrimination~\cite{Var85}, analysing the relationship between prices of broadband Internet access and time of use service~\cite{Var12}, and 
 auctions for advertisement position placement on page search results from Google~\cite{Var12,Var09}.


\subsection*{Perspective}

In
Bayesian estimation,  
the twin effects of social learning (information aggregation with interaction amongst agents) and data incest 
(misinformation propagation) lead to non-standard information patterns in estimating the underlying state of nature.
 Herding occurs when the public belief overrides the private observations and thus actions of agents are independent of their private observations.
Data incest results in bias in the public belief as a consequence of the unintentional re-use of identical actions in the formation of public belief in social learning; the information gathered by each agent is mistakenly considered to be independent. This results in overconfidence and bias in estimates of the state.

Privacy and reputation pose conflicting requirements: 
 privacy requirements  result in  noisier measurements or lower resolution actions (since individuals are not willing to disclose private observations), while a high degree
 of reputation requires accurate measurements.  Utility functions, noisy private measurements and quantized actions are essential
 ingredients of the  social  learning models presented in this paper 
that facilitate
modeling this tradeoff between reputation and privacy.


\section{Reputation Dynamics and Data Incest} 
   \label{sec:classicalsocial}

\subsection{Classical Social Learning}\label{sec:herd}

 
Consider a multi-agent system that aims to estimate the state of an underlying finite state random variable $x \in
\X = \{1,2,\ldots,X\}$ with known prior distribution $\pi_0$.
Each agent acts once  in a predetermined sequential order indexed by $k=1,2,\ldots$   
Assume at the beginning of iteration $k$,
all agents have access to the public belief $\pi_{k-1}$ defined in  Step (iv) below.
The social learning protocol proceeds as follows
 \cite{BHW92,Cha04}:\\
 (i) {\em Private Observation}: At time $k$,
agent $k$  records a private observation $y_k\in \Y $ 
from the observation distribution $B_{iy} = P(y|x=i)$, $i \in \X$.
Throughout this paper we assume that $\Y = \{1,2,\ldots,Y\}$ is finite.
\\
(ii) 
{\em Private Belief}:  Using the public belief $\pi_{k-1} $ available at time $k-1$ (Step (iv) below), agent $k$   updates its private
posterior belief  $\priv_k(i)  =  P(x_k = i| a_1,\ldots,a_{k-1},y_k)$ using Bayes formula:
\begin{align}  \label{eq:hmm} \priv_k &= 
\frac{B_{y_k} \pi}{ \mathbf{1}_X^\p B_y  \pi}, \;  
B_{y_k} = \text{diag}(P(y_k|x=i),i\in \X) . 
 \end{align}
 Here $\mathbf{1}_X$ denotes the $X$-dimensional vector of ones, $\eta_k$ is an $X$-dimensional probability mass function (pmf). \\
  (iii)   {\em Myopic Action}: Agent  $k$  takes  action $a_k\in \A = \{1,2,\ldots, A\}$ to  minimize its expected cost 
 \begin{multline}  
  a_k =  \arg\min_{a \in \A} \E\{c(x,a)|a_1,\ldots,a_{k-1},y_k\} \\  =\arg\min_{a\in \A} \{c_a^\p\priv_k\}.    \label{eq:myopic}
  \end{multline}
  Here $\ca = (c(i,a), i \in \X)$ denotes an $X$ dimensional cost vector, and $c(i,a)$ denotes the cost  incurred when the underlying state is $i$ and the  agent chooses action $a$.\\
(iv) {\em Social Learning Filter}:   
Given the action $a_k$ of agent $k$,  and the public belief $\pi_{k-1}$, each  subsequent agent $k' > k$ 
performs social learning to
update the public belief $\pi_k$ according to the  ``social learning  filter":\ \beq \pi_k = \Ts(\pi_{k-1},a_k), \text{ where } \Ts(\pi,a) = 
 \frac{\Bs_a \,\pi}{\sigs(\pi,a)}, \label{eq:piupdate}\eeq
where
$\sigs(\pi,a) = \mathbf{1}_X^\p \Bs_a \tp^\p \pi$ is the normalization factor of the Bayesian update.
In (\ref{eq:piupdate}),  the public belief $\pi_k(i)  = P(x_k = i|a_1,\ldots a_k)$ and $\Bs_a  = \text{diag}(P(a|x=i,\pi),i\in \X ) $ has elements
\begin{align*} 
 & P(a_k = a|x_k=i,\pi_{k-1}=\pi) = \sum_{y\in \Y} P(a|y,\pi)P(y|x_k=i) \\
 &   P(a_k=a|y,\pi) = \begin{cases}  1 \text{ if }  c_a^\p B_y \tp^\p \pi \leq c_{\ta}^\p B_y \tp^\p\pi, \; \ta \in \A \\
 0  \text{ otherwise. }  \end{cases}
  \end{align*}

 The following
result which is well known in the economics literature \cite{BHW92,Cha04}:

\begin{theorem}[\cite{BHW92}] 
\label{thm:herd} The social learning protocol  leads to an {\em information cascade}\footnote{
  A {\em herd of agents} takes place at time $\bar{k}$, if the actions of all agents after time $\bar{k}$ are identical, i.e., $a_k = a_{\bar{k}}$ for all
time  $k > \bar{k}$. An information cascade implies that a herd of agents occur.
 \cite{TBB10}  quotes the following anecdote of user  influence and herding in a social network: ``... when a popular blogger left his blogging site for a two-week vacation, the site's visitor tally fell, and content produced by three invited substitute bloggers could not stem the decline."}  in finite time
with probability~1. That is, after some finite 
 time $\bar{k}$ social learning ceases and the public belief $\pi_{k+1} = \pi_k$, $k \geq \bar{k}$, and all agents choose the same action  $a_{k+1} = a_k$, $k\geq \bar{k}$.
  \qed\end{theorem}

\subsection{Information Exchange Model} \label{sec:iem}
 In comparison to the previous subsection,
we now  consider social learning on a family of time dependent directed acyclic graphs  - in such cases, apart from herding, the phenomenon of data incest arises. 






 Consider an online reputation system  comprised of agents $\{1,2,\ldots,S\}$ that aim to estimate an underlying state of nature (a random variable). 
 Let $x \in 
 \X = \{1,2,\ldots,X\}$
 represent the state of nature (such as the quality of a restaurant/hotel) with known prior distribution $\pi_0$. Let $k = 1,2,3,\ldots$ depict epochs at which events occur. 
 The index $k$ marks the historical order of events and not  absolute time. For simplicity, we refer to $k$ as ``time".


 It is convenient also to reduce the coordinates of time $k$ and agent $s$  to a single integer index  $n$:
\begin{equation} \label{reindexing_scheme}
 n \triangleq s+ S(k-1), \quad
s \in \{1,\ldots, S\}, \; k = 1,2,3,\ldots 
\end{equation}
We  refer to $n$ as a ``node" of a time dependent information flow   graph $G_n$ which we now define.
 Let \begin{equation}\label{eq:defG} G_{n} = (V_{n}, E_{n}), \quad n  = 1,2,\ldots \end{equation} denote a sequence of time-dependent {\em directed acyclic graphs}
  ({\em DAGs})\footnote{A DAG is a directed graph with no directed cycles. The ordering of nodes $n=1,2,\ldots,$ proposed here is  a special case of the well known result
 that the nodes of a DAG are  partially orderable; see  Sec.\ref{sec:vote}.}
  of information flow in the social network until and including time $k$. Each vertex in $V_{n}$ represents an agent $s'$  at time $k'$ and each edge $(n',n'')$ in $E_{n}\subseteq V_{n} \times V_{n}$ shows that the information (action) of node $n'$ (agent $s'$ at time $k'$) reaches node $n''$ (agent $s''$ at time $k''$).  It is clear that  $G_n$ is  a sub-graph of $G_{n+1}$. 

The Adjacency Matrix $A_n$ of $G_n $ is an $n\times n$ matrix with elements $A_n(i,j)$  given by 
\beq \label{eq:adjacencymatrix}
A_n (i,j)=\begin{cases}
1 &\textrm{ if } (v_j,v_i)\in E \;, \\
0 &\textrm{ otherwise}
\end{cases}\;, \text{  } A_n(i,i)=0.
\eeq

The transitive closure matrix $T_n$ is the  $n\times n$ matrix 
\beq T_n =  \text{sgn}((\mathbf{I}_n-A_n)^{-1}) \label{eq:tc} \eeq
where for  matrix $M$, the matrix $\text{sgn}(M)$ has elements
$$ 
\text{sgn}(M)(i,j) = \begin{cases} 0 & \text{ if } M(i,j)=0\;, \\
1  & \text{ if } M(i,j) \neq 0. \end{cases}
$$
Note that $A_n(i,j) = 1$ if there is a single hop path between nodes $i$ and $j$, In comparison,
$T_n(i,j) = 1$ if there exists a path (possible multi-hop) between  $i$ and $j$.

The information reaching node $n$ depends on the information flow graph  $G_n$.
The following two sets will be used  to specify the incest removal algorithms below:
\begin{align}
\history_n &= \{m : A_n(m,n) = 1 \}  \label{eq:history}  \\
\full_n &= \{m : T_n(m,n) = 1 \} . \label{eq:full}  
\end{align}
Thus $\history_n$ denotes the set of previous nodes $m$ that communicate with node $n$ in a single-hop.
In comparison, $\full_n$
  denotes the set of previous nodes $m$ whose information eventually arrives at node $n$. Thus $\full_n$  contains all possible  multi-hop connections by which information from a node $m$
 eventually reaches node $n$. 

\subsubsection*{Properties of $A_n$ and $T_n$}
 Due to causality with respect to the time index $k$ (information sent by an agent can only arrive at another agent at a later time instant),
 the following obvious properties hold  (proof omitted):

\begin{lemma} \label{lem:properties} Consider the sequence of DAGs $G_n$, $n=1,2,\ldots$ \\
(i) The adjacency matrices $A_n$  are  upper triangular. \\  
 $A_n$ is the upper left  $n\times n$ submatrix of $A_{n+1}$. \\
(ii) The transitive closure matrices
 $T_n$  are upper triangular  with ones on the diagonal. Hence, $T_n$ is invertible.
\\
(iii)
Classical social learning of Sec.\ref{sec:herd}, is a trivial example with 
adjacency matrix
$A_n(i,j) = 1$ for $j=i+1$ and $A_n(i,j) = 0$ elsewhere. \qed
\end{lemma}

The appendix contains an  example of data incest that illustrates the above notation.

\subsection{Data Incest Model and Social Influence Constraint} \label{sec:incest}

Each node $n$ receives recommendations from its immediate friends (one hop neighbors) according to the information flow graph defined above.
 That is, it receives  actions  $\{a_m, m \in \history_n\}$ from  nodes  $m \in \history_n$ and then seeks 
to  compute the associated public beliefs  $\belief_m, m\in \history_n$.
If node $n$ naively (incorrectly)  assumes that the public beliefs $\belief_m, m\in \history_n$  are 
  independent, then
it would    
  fuse  these  as
  \beq  \belief_n{-} = \frac{\prod_{m\in \history_n} \belief_m }{\mathbf{1}_X^\p \prod_{m\in \history_n} \belief_m}.\label{eq:dataincest}\eeq
  This naive data fusion would result in data incest.
  
\subsubsection{Aim} The aim is to provide each node $n$ the true posterior distribution 
\beq \tbelief_{n-}(i) = P(x = i | \{a_m, m \in \full_n\})   \label{eq:aims}\eeq
subject to the following {\em social influence constraint}:
There exists a fusion algorithm $\mathcal{A}$ such that
  \beq  \tbelief_{n-} = \mathcal{A} (\pi_m, m \in \history_n).   \label{eq:si}\eeq

 \subsubsection{Discussion. Fair Rating and Social Influence}
We briefly pause to discuss 
  (\ref{eq:aims}) and (\ref{eq:si}).\\
(i) We  call $\tbelief_{n-}$ in (\ref{eq:aims})  the {\em true}  or {\em fair online rating} available to
  node $n$
 since 
  $\full_n $ defined in (\ref{eq:full}) denotes all  information (multi-hop links) available  to node $n$. By definition
  $\tbelief_{n-}$ is incest free and is the desired conditional probability that agent $n$ needs.
\footnote{%
  For the reader unfamiliar with Bayesian state estimation: Computing the posterior  $\tbelief_{n-}$ is crucial
 noisy observations.
 Then the
conditional mean estimate is evaluated as $\E\{x| \{a_m, m \in \full_n\} \}= \sum_{x \in \X} 
 x \tbelief_{n-}(x) $ and  is the minimum variance estimate, i.e., optimal in the mean square error sense and more generally
 in the Bregmann loss function sense. The conditional
mean is a `soft' estimate and is unbiased by definition.
  Alternatively the maximum aposteriori `hard' estimate is evaluated as $\argmax_x \tbelief_{n-}(x)$. \label{foot:bayes}}

  Indeed, if node $n$ combines  $\tbelief_{n-}$ together with its own private observation via social learning, then
clearly
\begin{align*}   \eta_n(i) &=  P(x = i | \{ a_m, m \in \full_n\},  y_n ), \quad i \in \X,  \\
\pi_n(i)  &= P(x = i | \{a_m, m \in \full_n\},  a_n ),\quad i \in \X, \end{align*}
are, respectively,   the correct (incest free) private belief for node $n$ and the correct after-action public belief.
If agent $n$ does not use  $\tbelief_{n-}$, then incest can propagate; for example if agent $n$ naively uses (\ref{eq:dataincest}).

Why should an individual $n$ agree to  use  $\tbelief_{n-}$ to combine with its private message? It is here that the social influence constraint (\ref{eq:si}) is important.
   $ \history_n$ can be viewed as the  ``{\em social message}'', i.e., personal friends of node $n$ since they directly communicate to node $n$, while the associated beliefs can be viewed as the
``{\em informational message}''.
As described in the remarkable recent paper  \cite{BFJ12}, the social message  from personal friends exerts a  large social influence\footnote{In a study conducted by social networking site {\em myYearbook}, 81\% of respondents said they had received advice from friends and followers relating to a product purchase through a social site; 74 percent of those who received such advice found it to be influential in their decision. ({\em Click Z}, January 2010).}
 -- it provides significant incentive (peer pressure) for individual $n$ to comply with
the protocol of combining its estimate with  $\tbelief_{n-}$  and thereby prevent incest.
 \cite{BFJ12} shows  that  receiving  messages 
from known friends has significantly more influence on an individual than the information in the messages. This study includes a comparison of information messages and social messages on Facebook and their direct
effect on voting behavior.
 To quote \cite{BFJ12}, ``The effect of social transmission on real-world voting
was greater than the direct effect of the messages themselves..."  In Sec.\ref{sec:expt}, we provide results of an experiment on human subjects that also  illustrates social influence in social learning. 

\subsection{Fair Online Reputation System: Protocol \ref{protocol:socialcons}} \label{sec:repute}

 The procedure specified in  Protocol \ref{protocol:socialcons}   evaluates the  fair online rating by eliminating
data incest in a social network. The aim is to achieve (\ref{eq:aims}) subject to (\ref{eq:si}).

\begin{algorithm}\floatname{algorithm}{Protocol}
(i) {\em  Information from Social Network}: 
\begin{compactenum}
\item {\em Social Message from Friends}: Node $n$ receives social message $\history_n$ comprising the names or photos of friends that have made  recommendations.
\item {\em Informational Message from Friends}: 
The reputation system    fuses  recommendations   $\{a_m, m \in \history_n\}$   into  the single informational message $\belief_{n-}$
and presents this to  node $n$.
\\
The reputation system computes $\belief_{n-}$  as follows:
\begin{compactenum}
\item   $\{a_m, m \in \history_n\}$ are used to compute public beliefs $\{\pi_m, m \in \history_n\}$ using Step (v) below.
\item 
$\{\pi_m, m \in \history_n\}$ are fused into  $\belief_{n-}$ as
\beq  \belief_{n-} = \mathcal{A} (\pi_m, m \in \history_n) . \label{eq:fuse}\eeq
In Sec.\ref{sec:removal},   fusion algorithm $\mathcal{A}$ is  designed as 
\beq   \lbelief_{n-}(i) =  \sum_{m  \in \history_n} w_n(m) \, \lbelief_m(i), \quad i \in \X.    \label{eq:algA} \eeq
Here $\lbelief_m (i) = \log \belief_m(i)$ and $w_n(m)$ are weights.
\end{compactenum}

\end{compactenum}
(ii) {\em  Observation}: Node $n$  records   private observation  $y_n$ from the  distribution $B_{iy} = P(y|x=i)$, $i \in \X$.\\
(iii) {\em Private Belief}:  Node $n$   uses $y_n$ and informational message $\belief_{n-}$ to update its private belief via Bayes rule:
\beq   \priv_n = 
\frac{B_{y_n}  \belief_{n-}}{ \mathbf{1}_X^\p B_y  \belief_{n-}} . \label{eq:privi} \eeq
\\
(iv) {\em Recommendation}: Node $n$ makes recommendation
$$a_n = \arg\min_a c_a^\p \eta_n$$ and records this on the  reputation system.\\
(v) {\em Public Belief Update by Network Administrator}: Based on recommendation $a_n$,
the reputation system (automated algorithm)  computes the public belief $\pi_n$ using  the social learning
filter (\ref{eq:piupdate}).
\caption{Incest Removal   in Online Reputation System}\label{protocol:socialcons}
\end{algorithm}

At this stage, the public rating $\belief_{n-}$ computed in (\ref{eq:fuse}) of  Protocol \ref{protocol:socialcons} is not necessarily the fair online rating  $\tbelief_{n-}$ of (\ref{eq:aims}).
Without careful design of  algorithm $\mathcal{A}$ in (\ref{eq:fuse}), 
due to  inter-dependencies of actions on previous actions,  $\belief_{n-}$  can be substantially different
from $\tbelief_{n-}$. 
Then  $ \eta_n $ computed via (\ref{eq:privi}) will not be the correct private belief and 
  incest will propagate in the network.  In other words,  $\eta_n$, $\belief_{n-}$ and $\belief_n$ are defined purely in terms of their computational expressions in Protocol \ref{protocol:socialcons};  
  they are not necessarily the desired  conditional probabilities, unless algorithm $\mathcal{A}$ is properly designed to remove incest.
 Note also the requirement that algorithm $\mathcal{A}$ needs to satisfy the social influence constraint (\ref{eq:si}).

 \subsection{Ordinal Decision Making in Protocol \ref{protocol:socialcons}} \label{sec:ordinal}

Protocol \ref{protocol:socialcons} assumes that each agent is a Bayesian utility optimizer.
The following discussion shows that under reasonable conditions, such a Bayesian model is a useful idealization of  agents' behaviors.

Humans typically make {\em monotone} decisions - the more favorable the private  observation,
the higher the recommendation. Humans   make {\em ordinal} decisions\footnote{Humans typically convert numerical attributes to ordinal scales before making  decisions. For example,
it does not matter if the cost of a meal at a restaurant is \$200 or \$205; an individual would classify this cost as ``high". 
Also credit rating agencies use ordinal symbols such as AAA, AA, A.} since humans tend to think in symbolic ordinal terms.
Under what conditions is the recommendation $a_n$ made by  node $n$ {\em monotone increasing} in its observation
$y_n$ and {\em ordinal}? 
Recall from Steps (iii) and (iv)  of Protocol \ref{protocol:socialcons} that the recommendation of node $n$ is
$$a_n(\tbelief_{n-},y_n) = \argmin_a c_a ^\p B_{y_n} \tbelief_{n{-}}$$
So an equivalent question is: Under what conditions is the $\argmin$ increasing in observation $y_n$?
Note that an increasing argmin is an {\em ordinal} property - that is,
$  \argmin_a c_a ^\p B_{y_n} \tbelief_{n{-}}$ increasing in $y$ implies $\argmin_a \phi(c_a ^\p B_{y_n} \tbelief_{n{-}})$ is also increasing in $y$ for any monotone function $\phi(\cdot)$.

The following result gives sufficient conditions for each agent to give a  recommendation that is monotone and ordinal in its private observation:
\begin{theorem} \label{thm:monotone}
Suppose the observation probabilities and costs satisfy the following conditions:
\begin{compactenum}
\item[(A1)] $B_{iy}$ are TP2 (totally positive of order 2); that is,
$B_{i+1,y}B_{i,y+1} \leq B_{i,y} B_{i+1,y+1}$.
\item[(A2)]  $c(x,a)$ is submodular. That is, $c(x,a+1) - c(x,a) \leq c(x+1,a+1)-c(x+1,a)$.
\end{compactenum}
Then
\begin{compactenum}
\item
Under (A1) and (A2),  the recommendation $a_n(\tbelief_{n-},y_n) $  made by agent $n$ is increasing and hence ordinal in observation $y_n$, for any $\tbelief_{n{-}}$. 
\item Under (A2),     $a_n(\tbelief_{n-},y_n) $ is increasing in belief $\tbelief_{n-}$ with respect to the monotone likelihood ratio (MLR) stochastic order\footnote{ \label{footnotemlr} Given probability mass functions
$\{p_i\}$ and $\{q_i\}$, $i=1,\ldots,X$ then
$p$ MLR dominates $q$ if  $\log p_i - \log p_{i+1} \leq \log q_i - \log q_{i+1}$.}  for any observation~$y_n$.\qed
\end{compactenum} 
\end{theorem}

The proof is in the appendix.
We can interpret the above theorem as follows. If agents makes recommendations that are monotone and ordinal in the observations and monotone in the prior, then they mimic
the Bayesian social learning model.  Even if the agent does not exactly follow a Bayesian social learning model,  its monotone ordinal behavior implies that such a Bayesian model is  a useful idealization.

Condition (A1) is widely studied in monotone decision making; see the classical book by  Karlin \cite{Kar68} and  \cite{KR80}; numerous examples of noise distributions are  TP2. Indeed in the highly cited paper
\cite{Mil81} in the economics literature, observation $y+1$ is said to be more ``favorable news''  than observation $y$ if Condition (A1) holds.

Condition (A2) is the well known submodularity condition \cite{Top98}. Actually (A2) is a stronger version of the   more general single-crossing condition \cite{MS92,Ami05} stemming from the economics
literature (see appendix)
$$ (c_{a+1} - c_a )^\p B_{y+1} \tbelief \geq 0 \implies (c_{a+1} - c_a )^\p B_{y} \tbelief \geq 0.$$ This single crossing condition is ordinal, since  for any monotone function $\phi$,
it is equivalent to
$$ \phi( (c_{a+1} - c_a )^\p B_{y+1} \tbelief ) \geq 0 \implies \phi( (c_{a+1} - c_a )^\p B_{y} \tbelief)  \geq 0.$$
(A2) also makes sense in a reputation system for the costs to be well posed.
 Suppose the recommendations in action set $\A$ are arranged
in increasing order and also  the states in $\X$ for the underlying state are arranged in ascending order. Then (A2) says:  if recommendation $a+1$  is more accurate
than recommendation $a$ for state $x$; 
then recommendation $a+1$ is also more accurate than recommendation $a$ for state $x+1$ (which is a higher quality state than $x$).

In the experiment  results reported in Sec.\ref{sec:expt}, we found that  (A1) and (A2) of Theorem \ref{thm:monotone} are justified.

\subsection{Discussion and Properties of Protocol \ref{protocol:socialcons}} \label{sec:discussion}

For the reader's convenience we provide an illustrative example of data incest in the appendix.
 We now discuss several other properties of Protocol~\ref{protocol:socialcons}.
 
 \subsubsection{Individuals have selective memory}  Protocol \ref{protocol:socialcons}  allows for  cases where 
each node can remember some (or all) of its past actions or none. This models cases where
people forget most of the past except for specific highlights.  For example, in the information flow graph of the illustrative example in the appendix (Fig.\ref{sample}),  if nodes 1,3,4 and 7 are assumed to be the same individual, then at node 7, the individual remembers what happened at node 5 and node 1, but not node 3.

\subsubsection{Security: Network  and Data Administrator} The social influence constraint (\ref{eq:si}) can  be viewed as a 
separation of privilege requirement for network security.  For example, the National Institute of Standards and Technology (NIST) of the U.S.\ Department of Commerce
 \cite[Sec.2.4]{SJT08} recommends  separating the  roles of data  and systems administrator.
Protocol \ref{protocol:socialcons} can then be interpreted as follows: A network administrator has access to the social network  graph (but not the data) and can compute weights 
$w_n(m)$ (\ref{eq:fuse})
by which the estimates are weighed.
A data administrator has access to the recommendations of friends (but not to the social network). Combining the
weights with the recommendations yields the informational message $\tbelief_{n-1}$ as in (\ref{eq:algA}).

\subsubsection{Automated Recommender System} Steps (i) and (v)  of Protocol \ref{protocol:socialcons} can be combined into an automated recommender system  that maps previous actions of agents in the social group
to a single recommendation (rating) $\belief_{n-}$ of (\ref{eq:fuse}). This recommender system can operate completely opaquely to the actual user (node $n$). Node $n$ simply
uses the automated rating $\belief_{n-}$ as the current best available rating from the reputation system. Actually Algorithm $\mathcal{A}$ presented below fuses the beliefs in a linear fashion. A human node $n$
 receiving an informational message comprised of a linear combination of recommendation of friends, along with the social message has incentive to follow the protocol as described in Sec.\ref{sec:incest}.

\subsubsection{Agent Reputation}
The  cost function minimization in Step (iv) can be interpreted in terms of the reputation of agents in online reputation systems. If an agent continues to write bad reviews for high quality
restaurants on Yelp, her reputation becomes lower among the users. Consequently, other people ignore reviews of that (low-reputation) agent in evaluating their opinions about the social unit under study (restaurant). Therefore, agents  minimize the penalty of writing inaccurate reviews.

 \subsection{Incest Removal Algorithm $\mathcal{A}$} \label{sec:removal}
 It only remains to  describe the construction of algorithm $\mathcal{A}$
in Step 2b of Protocol \ref{protocol:socialcons} so that 
\begin{align}
 \belief_{n-}(i) & =  \tbelief_{n-}(i) , \quad i \in \X \nonumber \\
 \text{ where } & \tbelief_{n-}(i) = P(x = i | \{a_m, m \in \full_n\}).   \label{eq:aimalg}\end{align}
To describe algorithm $\mathcal{A}$, we make the following definitions:
Recall  $\tbelief_{n-}$ in (\ref{eq:aims})  is the  {\em fair online rating} available to
  node $n$.
 It is convenient to work with the logarithm of the un-normalized belief:
   accordingly define
$$  \lbelief_n(i) \propto \log \belief_n(i), \quad \lbelief_{n-}(i) \propto  \log \belief_{n-}(i), \quad i \in \X.$$
Define the
  $n-1$ dimensional weight vector:
\beq  w_n =  T_{n-1}^{-1}  t_n.  \label{eq:weight} \eeq
Recall  that
$t_n$ denotes the first $n-1$ elements of the  $n$th column of the transitive closure matrix $T_n$. Thus the weights are purely a function of the adjacency matrix of
the graph and do not depend on the observed data.

We present algorithm $\mathcal{A}$ in two steps: first, the actual computation is given in Theorem \ref{thm:socialincestfilter}, second, necessary and sufficient conditions on the information flow graph for the existence of such an algorithm to achieve the social influence constraint (\ref{eq:si}).


  \begin{theorem}[Fair Rating Algorithm] \label{thm:socialincestfilter} Consider the  reputation system with  Protocol \ref{protocol:socialcons}.
Suppose the  network administrator runs the following algorithm in~(\ref{eq:fuse}):
\beq \label{eq:socialconstraintestimate}
\lbelief_{n-}(i)  =   \sum_{m=1}^{n-1} w_n(m)\, \lbelief_m(i)  
\eeq
where the weights $w_n$ are chosen according to (\ref{eq:weight}).

Then  $\lbelief_{n-}(i) \propto \log \tbelief_{n-}(i)$. That is,  the fair  rating $\log \tbelief_{n-}(i)$ defined in (\ref{eq:aims}) is obtained via (\ref{eq:socialconstraintestimate}).
 \qed
\end{theorem}

Theorem  \ref{thm:socialincestfilter}  says that  the fair rating  $\tbelief_{n-}$ can be expressed as a linear function of the action log-likelihoods
 in terms of the transitive closure matrix $T_n$ of graph $G_n$.

\subsubsection*{Achievability of Fair Rating by Protocol \ref{protocol:socialcons}}
\begin{compactenum}
\item Algorithm $\mathcal{A}$  at node $n$ specified by  (\ref{eq:fuse})  needs to satisfy the social influence constraint (\ref{eq:si}) - that is, it needs to operate on beliefs
$\lbelief_m, m \in \history_n$. 
\item 
On the other hand,  to provide incest free estimates, algorithm $\mathcal{A}$ specified in (\ref{eq:socialconstraintestimate})  requires  all previous beliefs $l_{1:n-1}(i)$ that are specified by the non-zero elements of the vector $ w_n$.
\end{compactenum}
The only way to reconcile points 1 and 2 is  to  ensure  that  $A_n(j,n) = 0$ implies $w_n(j) = 0$ for $j=1,\ldots, n-1$. This condition means that the single hop 
past estimates $\lbelief_m, m \in \history_n$
available at node $n$ according to  (\ref{eq:fuse}) in Protocol \ref{protocol:socialcons}  provide all the information required to compute 
$w_n^\p \, \lbelief_{1:n-1}$ in (\ref{eq:socialconstraintestimate}).
We formalize this condition in 
the following theorem.

\begin{theorem}[Achievability of Fair Rating]\label{thm:sufficient}
Consider the fair rating algorithm specified by (\ref{eq:socialconstraintestimate}). For Protocol \ref{protocol:socialcons}  using the social influence constraint information  $(\belief_m, m \in \history_n)$ to achieve the estimates $\lbelief_{n-}$ of algorithm (\ref{eq:socialconstraintestimate}),
a necessary and sufficient condition on the information flow graph $G_n$ is \beq \label{constraintnetwork}
A_n(j,n)=0   \implies w_n(j)= 0.
\eeq
Therefore for Protocol  \ref{protocol:socialcons} to generate incest free estimates for nodes $n=1,2,\ldots$, condition  (\ref{constraintnetwork}) needs to hold for each $n$.
(Recall 
  $w_n$ is specified in (\ref{eq:socialconstraintestimate}).)
\qed
\end{theorem}


\noindent{\em Summary}: Algorithm (\ref{eq:socialconstraintestimate}) together with the condition (\ref{constraintnetwork}) ensure that incest free estimates are generated
by Protocol \ref{protocol:socialcons} that satisfy social influence constraint~(\ref{eq:si}).

\section{Experimental Results on Human Subjects} \label{sec:expt}
To illustrate social learning, data incest and social influence, this section presents an actual 
psychology experiment that was conducted by our colleagues at the Department of Psychology of University of British Columbia in September and  October,  2013,
see \cite{HSK14} for details. The participants comprised 36 undergraduate students  who participated in the experiment for course credit.

\subsection{Experiment Setup}
The experimental study involved 1658 individual trials. Each trial comprised two participants who were asked to
 perform a perceptual task interactively.  
 The perceptual task was as follows:
 Two arrays of circles denoted left and right, were given to each pair of participants. Each participant was asked to judge which array (left or right) had the larger average diameter. The
 participants answer (left of right) constituted their action. So the action space is $\A = \{0 \text{ (left)} ,1 \text{ (right)}\}$.
 
The circles  were prepared for each trial as follows: two $4\times4$ grids of circles were generated by uniformly sampling  from the radii: $\{ 20, 24, 29, 35, 42\}$ (in pixels). The average diameter of each grid was computed, and if the means differed by more than 8\% or less than 4\%, new grids were made. Thus in each trial, the left array and right array
   of circles differed  in the average diameter  by 4-8\% 
   
 For each trial,
one of the two participants was chosen randomly to start the experiment  by choosing an action according to his/her observation. Thereafter, each participant was given access
to their partner's previous response (action) and the participants own previous action prior to making his/her  judgment. This mimics the social learning Protocol \ref{protocol:socialcons} of
 Sec.\ref{sec:repute}. 
  The participants continued choosing actions according to this procedure until the experiment terminated. The trial terminated when the response of each of the two participants did not change for three successive iterations (the two participants did not necessarily have to agree for the trial to terminate).
 
In each trial, the actions of participants were recorded along with the  time interval taken to choose their action. As an example, Fig.~\ref{Fig:Samplepath} illustrates the sample path of  decisions made by the two participants
in one of the 1658 trials. In this specific trial, the average diameter of the left array of circles was $32.1875$ and the right array was	$30.5625$ (in pixels); so the ground truth was $0 $ (left).

\begin{figure}[h]
\centerline{
\includegraphics[width=0.45\textwidth]{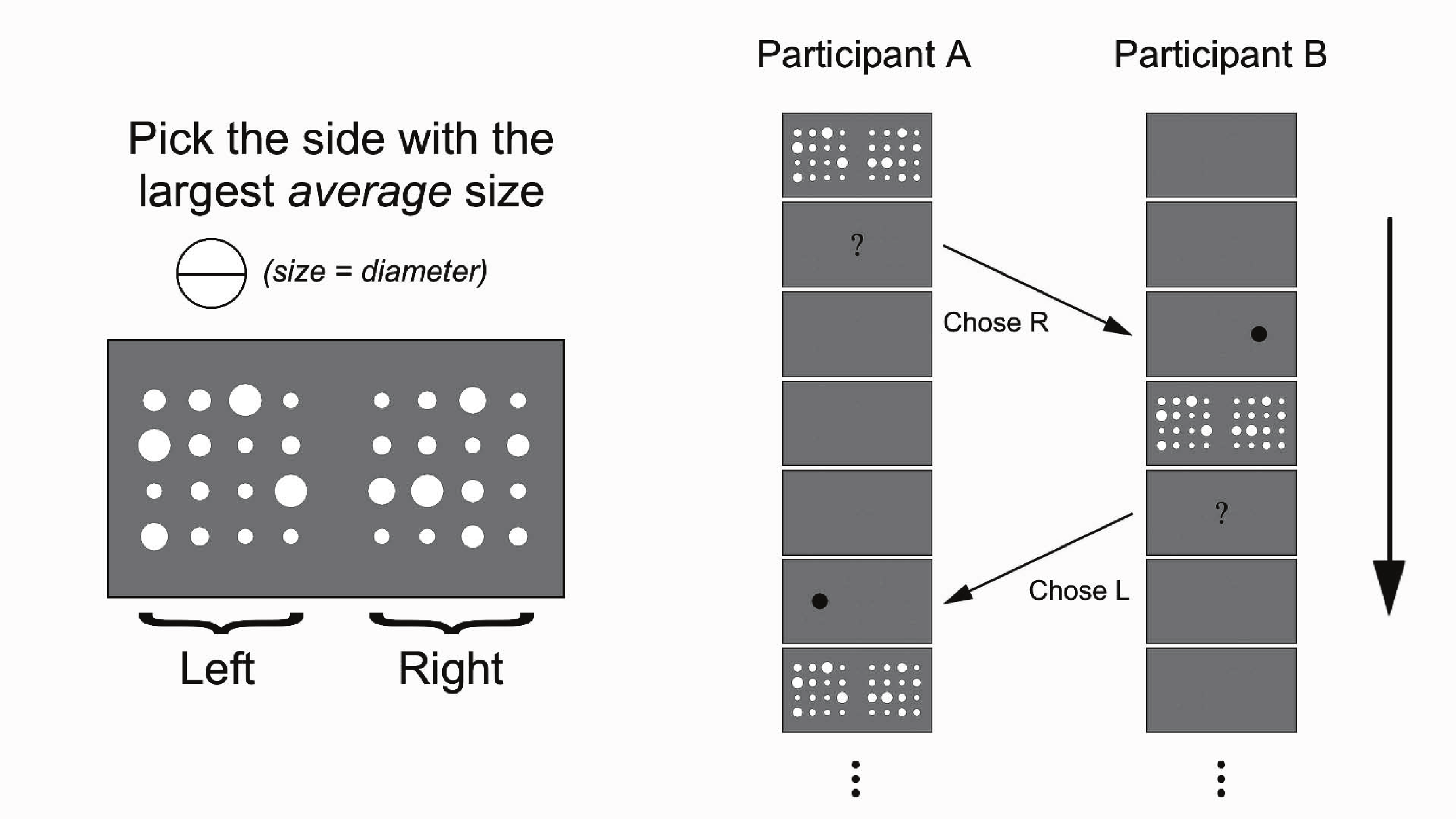}}
\caption{Two arrays of circles were given to each pair of participants on a screen. Their task is to interactively determine which side (either left or right) had the larger average diameter. The partner's previous decision was displayed on screen prior to the stimulus.}
\label{Fig:SocialSensor}
\end{figure}

 \begin{figure}[htb]
\centerline{
\includegraphics[width=0.4\textwidth]{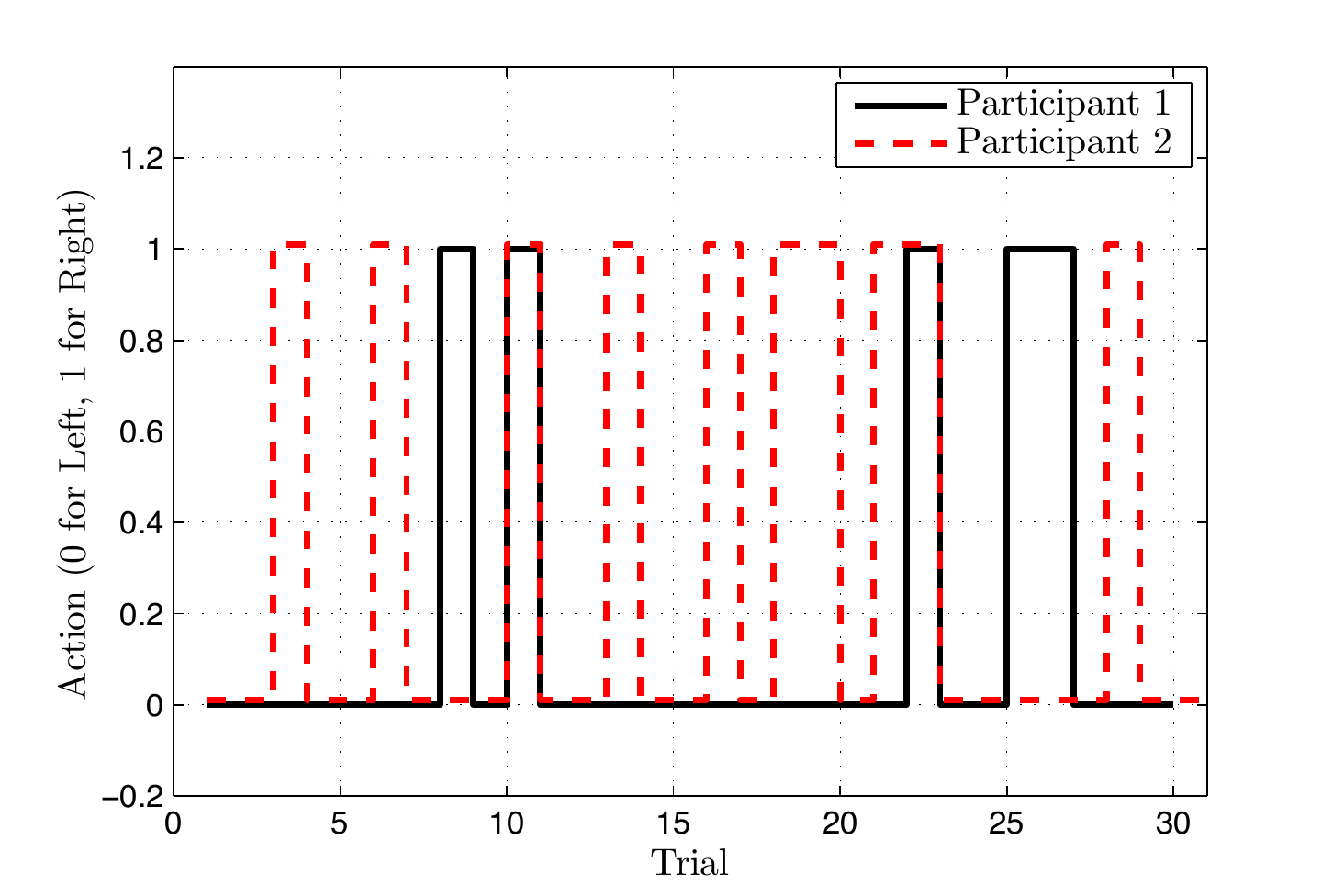}}
\caption{Example of sample  path of actions chosen by  two participants in a single trial of the experiment. In this  trial, both participants eventually chose the correct answer 0 (left).}
\label{Fig:Samplepath}
\end{figure}

\subsection{Experimental Results}   The results of our experimental study are as follows:

\subsubsection{Social learning Model} As mentioned above, the experiment for each pair of participants was continued until both participants' responses stabilized. 
 {\em In what percentage of these experiments, did an agreement occur between the two participants?} 
The answer to this question reveals  whether ``herding" occurred in the experiments
 and whether
the participants performed social learning (influenced by their partners). The experiments show that in  66\% of trials (1102 among 1658), participants reached an agreement; that is herding occurred. Further, in 32\% of the trials, both participants converged to the correct decision after a few interactions. 

To construct  a  social learning model for the experimental data,  we consider  the experiments where both participants reached an agreement.
Define the social learning success rate as
    \begin{equation}\nonumber
\frac{\text{\# expts where both participants chose correct answer}}{\text{\# expts where both participants reached an agreement}}\cdot
    \end{equation}
   In the experimental study, the state space is $\X = \{0,1\}$ where $x = 0$, when the left array of circles has the larger diameter and $x = 1$, when the right array has the larger diameter. The initial belief for both participants is considered to be $\pi_0 = [0.5, 0.5]$. The observation space is assumed to be $\Y = \{0,1\}$. 
   
   To estimate the social learning model parameters (observation probabilities  $B_{iy}$  and  costs  $c(i,a)$), 
we determined the parameters that best fit the learning success rate of the experimental data. The best fit parameters obtained were\footnote{Parameter estimation in social learning is a challenging problem not addressed in this paper. Due to the formation of  cascades in finite time, construction of an asymptotically  consistent estimator is impossible, since actions after the formation of a cascade contain no information.}
   \begin{align}
   &B_{iy}  = \begin{bmatrix} 0.61 & 0.39 \\ 0.41 & 0.59 \end{bmatrix}, \quad
   c(i,a) = \begin{bmatrix} 0 & 2 \\ 2 & 0 \end{bmatrix}\nonumber.  
   \end{align}
Note that $B_{iy}$ and $c(i,a)$ satisfy both the conditions of the Theorem \ref{thm:monotone}, namely TP2 observation probabilities and single-crossing cost. This implies that the subjects of this experiment made monotone and ordinal decisions.
   
\subsubsection{Data incest} Here, we study the effect of information patterns in  the experimental study that can result in data incest. Since private observations
are highly subjective and participants  did not
document these, we cannot claim with certainty if data incest changed the action of an individual. However, from the experimental data, we can localize specific
information patterns that can result in incest. In particular, we focus on  the two information flow graphs depicted in Fig.\ref{exp:dataincest}.
 In the two graphs of Fig.\ref{exp:dataincest}, the action of the first participant at time $k$ 
 influenced the action of the second participant at time $k+1$, and thus, could have been double
 counted by the first participant at time $k+2$. 
We found that in 
   79\% of experiments, one of the information patterns shown in Fig.\ref{exp:dataincest} occurred (1303 out of  1658 experiments). Further, in 21\% of experiments, 
  the information patterns shown in Fig.\ref{exp:dataincest} occurred and at least one participant
changed his/her decision, i.e., the judgment of participant at time $k+1$ differed from his/her judgments at time $k+2$ and $k$.  These results
show that even for   experiments involving  two participants, data incest information  patterns occur frequently (79\%) and causes individuals to modify their
actions (21\%). It shows that  social learning protocols require careful design to handle and mitigate data incest.
\begin{figure}[h]
\centering
\scalebox{.42}{\includegraphics{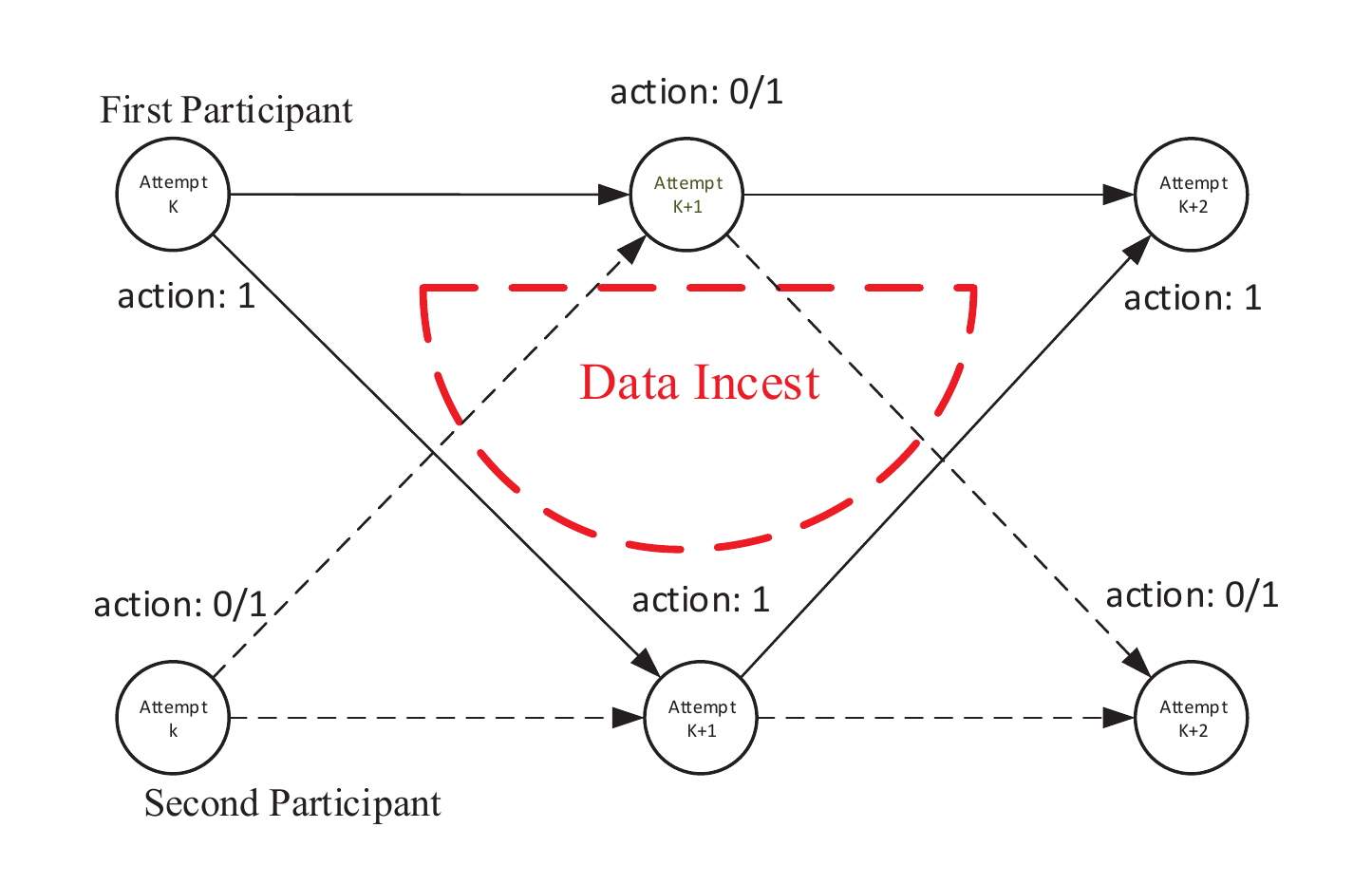}}

\scalebox{.42}{\includegraphics{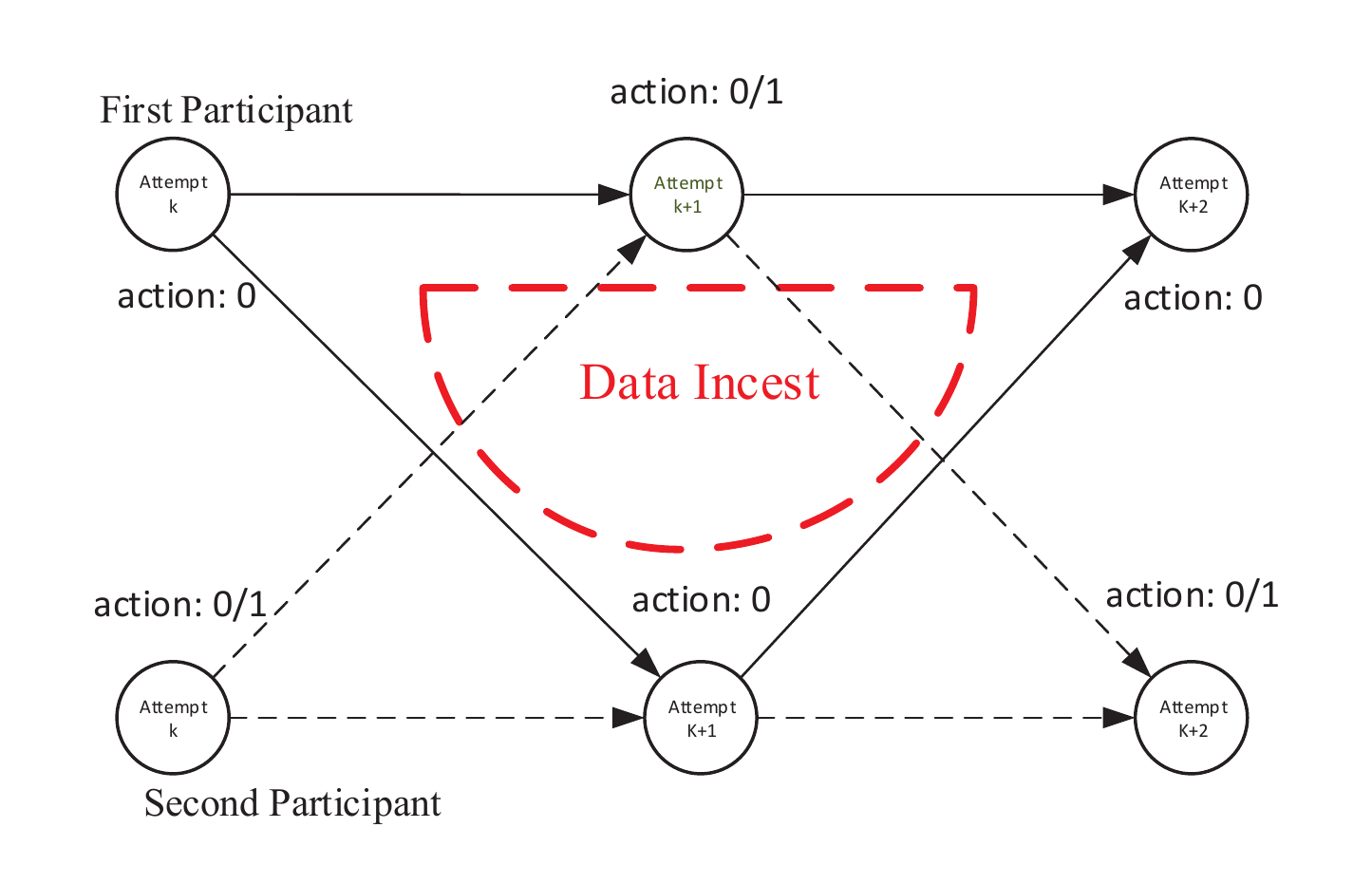}}
\caption{\label{exp:dataincest} Two information  patterns from our experimental studies  which can result in data incest.}
\end{figure}


\section{Belief-based Expectation Polling} \label{sec:vote}

We now move on to the third part of the paper, namely data incest in expectation polling. Unlike previous sections, the agents no longer are assumed to  eliminate incest at each step.
So incest propagates in the network.  Given the incestious beliefs, the aim is to compute the posterior of the state.
We illustrate how the  results of  Sec.\ref{sec:classicalsocial} can  be used to eliminate data incest (and therefore bias) in expectation polling systems.
Recall from Sec.\ref{sec:introduction}  that  in expectation  polling \cite{RW10,DKS12},  individuals are sampled and  asked who they believe will win the election; as opposed to intent polling where
individuals are sampled and asked who they intend to vote for. The  bias of the estimate from expectation polling
depends strongly on the social network structure.  Our approach below is  Bayesian:  we 
compute the posterior and therefore the conditional mean estimate  (see Footnote \ref{foot:bayes}).

 We consider two formulations below:
\begin{compactenum}
\item An  expectation polling system where in addition to specific polled voters, the minimal number of additional voters are recruited. The pollster then is able to use the incestious beliefs
to compute the posterior conditioned on the observations
 and  thereby eliminate incest (Sec.\ref{sec:polldesign} below).
\item An expectation polling system when it is not possible to recruit additional voters. The pollster then can only compute the posterior conditioned on the incestious beliefs.
(Sec.\ref{sec:pollestimate} below). 
\end{compactenum}

The first approach can be termed as ``exact" since the estimate computed based on the incestious beliefs is equivalent
to the estimate computed based on the private observations of the sampled voters. The second approach, although optimal given the available information, has a higher variance (see footnote \ref{foot:price}).

Suppose $X$ candidates contest an election.
Let $x \in 
 \X = \{1,2,\ldots,X\}$ denote the candidate that is leading  amongst the voters, i.e., $x$ is the true state of nature.  There are $N$ voters. These voters  communicate their expectations via a social network
  according to the steps listed in Protocol~\ref{polling}.
We index the voters as nodes $n \in \{1,\ldots, N\}$ as follows:
 $m< n$ if there exists a directed path from node $m$ to node $n$ in the directed acyclic graph (DAG).  It is well known that such a labeling of nodes in a DAG constitute a partially ordered set.
 The remaining nodes can be  indexed arbitrarily, providing the above partial ordering holds.

Protocol \ref{polling}  is similar to Protocol \ref{protocol:socialcons}, except that voters exchange their expectations (beliefs) of who the leading candidate is, instead of
recommendations. (So unlike the  previous section, agents
do not perform social learning.)

\begin{algorithm}\floatname{algorithm}{Protocol}
\begin{compactenum}
\item {\em Intention from Friends}: Node $n$ receives the beliefs $\{\pi_m, m \in \history_n\}$ of who is leading from its immediate friends, namely  nodes  $\history_n$.
\item {\em Naive Data Fusion}: Node $n$ fuses the estimates of its friends naively (and therefore with incest) as 
$$\belief_{n-} = \frac{\prod_{m\in \history_n} \belief_m }{\mathbf{1}_X^\p \prod_{m\in \history_n} \belief_m}.  $$
\item {\em  Observation}: Node $n$  records its    private observation  $y_n$ from the  distribution $B_{iy} = P(y|x=i)$, $i \in \X$ of who the leading candidate is.
\item {\em  Belief Update}:  Node $n$   uses $y_n$  to update its  belief via Bayes formula:
\beq   \belief_n = 
\frac{B_{y_n}  \belief_{n-}}{ \mathbf{1}_X^\p B_{y_n}  \belief_{n-}} . \label{eq:privpoll} \eeq
\item Node $n$ sends its  belief $\belief_n$ to subsequent nodes as specified by the social network graph.
\end{compactenum}
\caption{Belief-based Expectation Polling Protocol} \label{polling}
\end{algorithm}

{\em Remark}: Similar to Theorem \ref{thm:monotone}, the Bayesian update (\ref{eq:privpoll}) in Protocol \ref{polling}
can be viewed as an idealized model;  under   assumption (A1)
 the belief $\belief_n$  increases with respect to the observation $y_n$ (in terms of the monotone likelihood ratio order). So even  if agents are not actually
 Bayesian, if they choose their belief to be monotone
 in the observation, they mimic the Bayesian update.

\subsection{Exact Incest Removal in Expectation Polling} \label{sec:polldesign}

Assuming that the voters follow Protocol \ref{polling}, the polling agency (pollster) seeks  to estimate  the leading candidate (state of nature $x$) by sampling
a subset of the 
 $N$  voters.  
 
 \subsubsection{Setup and Aim} 
Let $\Req = \{\Req_1,\Req_2,\ldots,\Req_L\}$ denote   the   $L-1$ sampled {\em recruited}  voters together with node  $\poll$ that denotes the pollster.
For example the  $L-1$ voters could be volunteers or
paid  recruits  that have already signed up
for the polling agency.  Since, by Protocol \ref{polling}, these voters have naively combined the intentions of their friends (by ignoring
dependencies), the pollster needs to poll additional voters to eliminate incest. Let $\additional$ denote this set of {\em extra} (additional) voters to poll.
Clearly the choice of $\additional$  will  depend   on $\Req$ and the structure of the social network.

What is  the smallest set $\additional$ of extra voters  to poll in order to  compute the posterior $P(x|y_1,\ldots,y_{\poll})$ of state $x$ (and therefore eliminate data incest)?

\subsubsection{Incest Removal}
Assume that all the recruited $L-1$ nodes in $\Req$ report to a central polling node  $\poll$. We assume that the pollster has complete knowledge of the network;
e.g.,  in an online social network like Facebook.

Using the formulation of Sec.\ref{sec:iem},
the pollster constructs the directed acyclic graph $G_n = (V_n,E_n)$, $n\in \{1,\ldots,  \poll\}$.
The methodology of Sec.\ref{sec:removal} is straightforwardly used to determine the minimal set of extra (additional) voters $\additional$.
The procedure is as follows:
\begin{compactenum}
\item
For each $n= 1,\ldots, \poll$,   compute the weight vector $w_n = T_{n-1}^{-1}  t_n$;  see  (\ref{eq:weight}).
\item The indices of the non-zero elements in $w_n$ that are not in $V_n$, constitute the minimal additional  voters (nodes) that need to be polled.
This is due to the necessity and sufficiency of (\ref{constraintnetwork}) in Theorem \ref{thm:sufficient}.
\item Once the additional voters in $\additional$ are polled for their beliefs, the polling agency corrects belief $\belief_n$ reported by node $n \in \Req$ to remove incest as:
\beq  \cbelief_{n}(i)  =    \lbelief_n(i) +  \sum_{m \in \additional, m < n} w_n(m) \lbelief_m(i) . \label{eq:voteincestremove} \eeq
\end{compactenum}

\subsection{Bayesian Expectation Polling  using Incestious Beliefs}  \label{sec:pollestimate}

Consider  Protocol \ref{polling}. 
However, unlike  the previous subsection, assume that the pollster {\em cannot} sample additional voters $\additional$ to remove incest; e.g.,
the pollster is unable to  reach  marginalized sections of  the population. 

Given the incest containing beliefs $\belief_{\Req_1}, \ldots, \belief_{\poll}$ of the $L$ recruits and pollster, how can the pollster compute the posterior
$P(x|\belief_{\Req_1}, \ldots, \belief_{\poll})$ and therefore the unbiased conditional mean estimate $\E\{x | \belief_{\Req_1}, \ldots, \belief_{\poll}\}$ where
$x$ is the state of nature\footnote{The price to pay for  not recruiting additional voters is an increase in variance of the Bayesian estimate. It is  clear that $E\{x|y_1,\ldots,y_{\poll}\}$ computed using additional voters in Sec.\ref{sec:polldesign} has a lower variance (and is thus a more accurate estimator) than  $\E\{x | \belief_{\Req_1}, \ldots, \belief_{\poll}\}$ since the sigma-algebra generated by $ \belief_{\Req_1}, \ldots, \belief_{\poll}$ is a subset of that generated by $y_1,\ldots,y_{\poll}$. Of course, any conditional mean estimator is  unbiased by definition. \label{foot:price}}?

\subsubsection{Optimal  (Conditional Mean) Estimation using Incestious Beliefs}

Define the following  notation
\begin{align*}  \belief_{\Req} &=  \begin{bmatrix} \belief_{\Req_1} &  \ldots & \belief_{\Req_L} \end{bmatrix}^\p, \;
 \lbelief_{\Req} =  \begin{bmatrix} \lbelief_{\Req_1} &  \ldots & \lbelief_{\Req_L} \end{bmatrix}^\p,
\\ Y_{\Req} &= \begin{bmatrix} y_{1} & \ldots & y_{\Req_L} \end{bmatrix}^\p, \;
\logoprob(x) = \begin{bmatrix}  \log b_{xy_1} \ldots,  \log b_{xy_n}  \end{bmatrix}^\p
 \end{align*}

\begin{theorem} \label{thm:estpollster}
Consider  the beliefs $\belief_{\Req_1}, \ldots, \belief_{\Req_L}$ of the $L$ recruits and pollster in an expectation poll operating according to Protocol  \ref{polling}. Then for each $x \in \X$,
 the
posterior is evaluated as
\beq \label{eq:pollest}
P(x| \belief_{\Req}) \propto \sum_{Y_{\Req}  \in \obseq}  \prod_{m=1}^{\Req_L}  B_{x y_m}\, \pi_0(x) .
\eeq
Here  $\obseq$ denotes the set of sequences $\{Y_{\Req}\} $ that satisfy
\begin{align} \label{eq:omat}
  \omat  \, \logoprob(x)  & = \lbelief_\Req(x)  - \omat e_1 \lbelief_0(x)  \\
\text{ where, } \omat &= \begin{bmatrix} 
e _{\Req_1} &  e_{\Req_2} & \cdots &  e _{\Req_L}  \end{bmatrix}^\p
(I -  A')^{-1}. \nonumber
\end{align}
Recall $A$ is the adjacency matrix and $e_m$ denotes the unit $\Req_L$ dimension vector    with 1 in the $m$-th position. \qed
\end{theorem}

The above  theorem asserts that given the incest containing beliefs of the $L$ recruits and pollster, the posterior distribution of the
candidates can be computed via (\ref{eq:pollest}). 
The conditional mean estimate or maximum aposteriori estimate  can then be computed as in footnote \ref{foot:bayes}.

\subsection{An extreme example of Incest in Expectation Polling}
The following  example is a Bayesian version of polling in a social network  described  in \cite{DKS12}. We show that due to data incest,  expectation polling can be significantly biased.
 Then the methods of Sec.\ref{sec:polldesign} and \ref{sec:pollestimate} for eliminating incest are illustrated.

Consider the  social network
 Fig.\ref{fig:polling}  with  represents an expectation polling system. 
The $L-1$ recruited  nodes are denoted as $\Req = \{2,3,\ldots, L\}$. These sampled nodes   report their beliefs to the polling node $ L+1$.
Since the poll only considers sampled voters, for notational convenience, we ignore labeling the remaining  voters; therefore
Fig.\ref{polling} only shows  $L+1$ nodes.
 
 An important feature of the graph of 
Fig.\ref{fig:polling} is that all recruited nodes are influenced by  node 1.
%
%
This  unduly affects the estimates reported by every other node.
For large $L$,  even though node 1 constitutes a negligible fraction of the total number of  voters, it significantly biases the estimate  of $x$ due to incest.\footnote{In \cite{DKS12} a  novel
 method is proposed to achieve  unbiased expectation polling  - weigh the estimate of each node by the reciprocal of its degree distribution, or alternatively sample
nodes with probability inversely proportional to their degree distribution. Then highly influential
nodes such as node 1 in  Fig.\ref{fig:polling}  cannot bias the estimate. Our paper is motivated by Bayesian
considerations where we are interested in estimating the optimal (conditional mean) estimate which by definition is unbiased.}

To illustrate an extreme case of bias due to data incest, suppose that
 $\X = \{1,2\}$ (so there are $X=2$ candidates). Consider   Fig.\ref{fig:polling} with a total of $L=6$  recruited  nodes; so $\Req = \{2,\ldots,8\}$.
 Assume the private observations recorded at the  nodes $1,2,\ldots, 8 $ are, respectively,
$[2,1,1,1,1,1,1,2]$. Suppose the true state of nature is $x=1$, that is, candidate 1 is leading the poll.
The nodes exchange and compute  their beliefs  according to Protocol \ref{polling}.
For prior $\pi_0 = [0.5, 0.5]^\p$ and observation matrix $B = 
\begin{bmatrix} 0.8 & 0.2 \\ 0.2 & 0.8
\end{bmatrix}$,
 it is easily verified that $\belief_8(1) = 0.2$, $\belief_8(2)  = 0.8$. That is, even though all 6  samples recorded  candidate 1 
as winning (based on  private observations)  and the true state is 1, the belief $\belief_8$ of  the pollster  is significantly biased towards  candidate~2.

Next, we  examine the   two methods of incest removal proposed in Sec.\ref{sec:polldesign} and Sec.\ref{sec:pollestimate}, respectively.
Due to the structure of adjacency matrix $A$  of the network in Fig.\ref{fig:polling}, condition
(\ref{constraintnetwork}) does not hold and therefore exact  incest removal  is not possible unless node 1 is also sampled. Accordingly,
suppose node 1 is sampled in addition to nodes $\{2,\ldots,7\}$ and data incest is removed via algorithm (\ref{eq:voteincestremove}).
Then the incest free estimate is $P(x=1|y_1,\ldots,y_8) = 0.9961$.

 Finally,  consider the case where node 1 cannot be sampled.
 Then using (\ref{eq:pollest}), the posterior is computed as   $P(x=1|\belief_8) = 0.5544$.

Comparing the three estimates for candidate 1, namely, 0.2 (naive implementation of Protocol \ref{polling} with incest),  0.9961 (with optimal incest removal), and 0.5544 (incest removal based on $\belief_8)$, one can see that naive expectation polling is significantly biased -- recall that the ground  truth is that candidate 1 is the winning candidate.

\begin{figure}[h]
\centering
\begin{tikzpicture}[->,>=stealth',shorten >=1pt,auto,node distance=2cm,
  thick,main node/.style={circle,fill=white!12,draw,font=\sffamily},scale=0.7, every node/.style={transform shape}]
 
  \node[main node] (1) {1};
  \node[main node] (2) at (2,2) {$2$};
  \node[main node] (3) at (2,1){$3$};
   \node[main node] (4) at (2,0){$4$};
    \node (5) at (2,-1){{$\vdots$}};
     \node[main node] (6) at (2,-2){$L$};
   \node[main node] (7) at (4,0) {$\poll$};
  \path[every node/.style={font=\sffamily\small}]
    (1) edge node {} (2)
         edge node {} (3)
          edge node {} (4)
          edge node {} (5)
          edge node {} (6)
         (2) edge node {} (7)
           (3) edge node {} (7)
             (4) edge node {} (7)
               (5) edge node {} (7)
                 (6) edge node {} (7);
   \end{tikzpicture}
\caption{Expectation Polling where social network structure can result in significant bias. Node 1 has  undue influence on the beliefs of  other nodes. Node $\poll =  L+1$ represents
the pollster. Only sampled voters are shown.}
\label{fig:polling}
\end{figure}
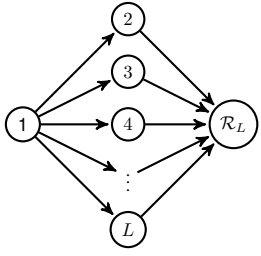

\section{Revealed Preferences and Social Learning} \label{sec:revealed}

A key assumption in the social learning models formulated in this paper is that agents are utility maximizers.
In microeconomics, the principle of revealed preferences seeks to determine if an agent is an utility maximizer subject to budget constraints
based on observing its  choices over time.   In this section we will use the principle of revealed preferences
on Twitter datasets  to illustrate social learning dynamics in reputation agencies  as a function of external influence (public belief).

\subsection{Afriat's Theorem}
\label{subsec:afriatstheorem}

 Given a time-series of data $\dataset=\{(\probe_\tindx,\response_\tindx), \tindx\in\{1,2,\dots,\Tindxter\}\}$ where $\probe_\tindx\in\reals^m$ denotes the public belief\footnote{In this section, the public belief is no longer a probability mass function. Instead it reflects the public's perception
 of the reputation agency
 based on the number of followers and tweet sentiment.}, $\response_\tindx$ denotes the response of agent, and $\tindx$ denotes the time index, is it possible to detect if the
   agent is a {\it utility maximizer}?
An agent is a {\em utility maximizer} at each time $\tindx$ if for every public belief $\probe_\tindx$, the chosen response $\response_\tindx$ satisfies
\begin{equation}
\response_\tindx(\probe_\tindx)\in\operatorname*{arg\,max}_{\{\probe_\tindx^\p \response \leq \budget_\tindx\}}\utility(\response)
\label{eqn:singlemaximization}
\end{equation}
with $\utility(\response)$ a non-satiated utility function. Nonsatiated means that an increase in any element of response $\response$ results in the utility function increasing.\footnote{The non-satiated assumption rules out  trivial cases such as a constant utility function which can be  optimized by any response.} As shown by Diewert~\cite{Die73}, without local nonsatiation the maximization problem (\ref{eqn:singlemaximization}) may have no solution. 

 In (\ref{eqn:singlemaximization}) the budget constraint $\probe_\tindx^\p \response_\tindx \leq \budget_\tindx$ denotes the total amount of resources available to the social sensor for selecting the response $x$  to the public belief $\probe_t$. In Sec.\ref{sec:twitterreview}, we will interpret this as a social
 impact budget.

Afriat's theorem \cite{Afr67,Afr87} provides a necessary and sufficient
  condition for a finite dataset $\dataset$ to have originated from an utility maximizer. 
\begin{theorem}[Afriat's Theorem]{Given a dataset $\mathcal{D}=\{(\probe_t,\response_t):t\in \{1,2,\dots,T\}\}$, the following statements are equivalent:}
\begin{compactenum}
  \item The agent is a utility maximizer and there exists a nonsatiated and concave utility function that satisfies (\ref{eqn:singlemaximization}).
\item For scalars $u_t$ and $\lambda_t>0$ the following set of inequalities has a feasible solution:
\begin{equation}
\utility_\tau-\utility_\tindx-\lambda_\tindx \probe_\tindx^\p (\response_\tau-\response_\tindx) \leq 0 \text{ for } \tindx,\tau\in\{1,2,\dots,\Tindxter\}.\
\label{eqn:AfriatFeasibilityTest}
\end{equation}
\item A nonsatiated and concave utility function that satisfies (\ref{eqn:singlemaximization}) is given by:
\begin{equation}
\utility(\response) = \underset{\tindx\in T}{\operatorname{min}}\{u_\tindx+\lambda_\tindx \probe_\tindx^\p(\response-\response_\tindx)\}
\label{eqn:estutility}
\end{equation}
  \item The dataset $\mathcal{D}$ satisfies the Generalized Axiom of Revealed Preference (GARP), namely for any $k\leq T$, $\probe_t^\p \response_t \geq \probe_t^\p \response_{t+1} \quad \forall t\leq k-1 \implies \probe_k^\p  \response_k \leq \probe_k^\p  \response_{1}.$
\end{compactenum}  \qed
\label{thrm: Afriat's Theorem}
\end{theorem}

As pointed out in  \cite{Var82}, an interesting feature of Afriat's theorem is that if the dataset can be rationalized by a non-trivial utility function, then it can be rationalized
by a continuous, concave, monotonic utility function.  That is, violations of continuity, concavity, or monotonicity cannot be detected with  a finite number of  observations.

Verifying  GARP  (statement 4 of Theorem \ref{thrm: Afriat's Theorem}) on a dataset $\dataset$ comprising $T$ points can be done using Warshall's algorithm with $O(\Tindxter^3)$~\cite{Var82,FST04} computations. Alternatively, determining if Afriat's inequalities (\ref{eqn:AfriatFeasibilityTest}) are feasible can be done via a LP feasibility test (using for example interior point methods \cite{BV04}). Note that the utility (\ref{eqn:estutility}) is not unique and is ordinal by construction. Ordinal means that any monotone increasing  transformation of the utility function will also satisfy Afriat's theorem. Therefore the utility mimics the ordinal behavior of humans, see also Sec.\ref{sec:ordinal}.
Geometrically the estimated utility (\ref{eqn:estutility}) is the lower envelop of a finite number of hyperplanes that is consistent with the dataset~$\dataset$.

\subsection{Example: Twitter Data of Online Reputation Agencies}
\label{sec:twitterreview}

In this section we illustrate  social learning associated with reputation agencies on  Twitter.
Content (games, movies, books) is provided to the reputation agency which then publishes its review as a  tweet. 
The sentiment of  the tweet and number of followers (reputation) constitute the publicly available
 information that drives the social learning process.
The Twitter network  responds to the sentiment of the tweet by retweeting.
 Data incest information structures arise since if the sentiment of the tweet published by the reputation agency improves its reputation, 
  then this reinforces the reputation agency's belief that it adequately reviewed the content.

The framework, which is   illustrated in Fig.\ref{fig:BodyConnect},
involves utility maximization and dynamics of the public belief, which we will show (based on Twitter data) evolves according to an autoregressive process.
  Specifically the goal is to investigate how the number of {\it followers} and {\it polarity}\footnote{Here polarity~\cite{MMMU14} is a real
valued variable in the interval [-1,1] and depends on the sentiment of the tweet; that is, whether the tweet expresses a positive/negative/neutral statement.} affect the time before a retweet occurs and the total number of associated retweets.  Apart from its relevance to social learning, the information provided by this analysis can be used in social media marketing strategies to improve a brand and for brand awareness. As discussed in \cite{GOT13}, Twitter provides a significant amount of agent-generated data which can be analyzed to provide novel personal advertising to agents. 

 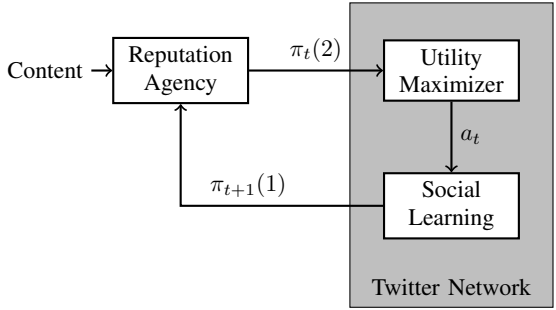
\begin{figure}[h]
  \centering
\begin{tikzpicture}[font = \normalsize, scale =0.9,transform shape, american voltages]

\draw [fill= lightgray] (1.5,-3.5) coordinate (topleft) rectangle (4.5,1) coordinate (bottomright);
\node at (3,-3.2) {Twitter Network};
\node (CN) at (-3.0,0) {Content};
\node[draw,rectangle,fill=white,thick,text width=5em,align=center] (RA) at (-1.0,0) {Reputation Agency}; 
\node[draw,rectangle,fill=white,thick,text width=5em,align=center] (UM) at (3,0) {Utility Maximizer}; 
\node[draw,rectangle,fill=white,thick,text width=5em,align=center] (SL) at (3,-2) {Social Learning}; 

\draw[thick, black, ->] (CN.east) --  (RA.west);
\draw[thick, black, ->] (RA.east) -- node[midway, above] {$\probe_\tindx(2)$} (UM.west);
\draw[thick, black, ->] (UM.south) -- node[midway, right] {$\response_\tindx$} (SL.north);
\draw[thick, black, ->] (SL.west) -| node[midway, above,xshift=1.0cm] {$\probe_{\tindx+1}(1)$} (RA.south);

\end{tikzpicture}
  \caption{Schematic of dynamics of the reputation agency and the Twitter network. 
  The reputation agency receives {\it content} (books, games, movies) which it reviews and publishes  tweets at epochs $\tindx=1,2,\ldots,$ with sentiment  $\probe_\tindx(2)$. The reputation agency index $i$ has been omitted for clarity.
  The public belief $\probe_\tindx$ and response $\response_\tindx$ are defined in Sec.\ref{sec:twitterreview}. } 
\label{fig:BodyConnect}
\end{figure}

\subsubsection{Twitter Datasets}
We consider 9 well known online reputation agencies: @IGN, @gamespot, @AmznMovieRevws, @creativereview, @HarvardBiz, @techreview, @pcgamer, @RottenTomatoes, @LAReviewofBooks.
Fig.\ref{fig:reviewnet} provides the social retweet network which includes the Twitter accounts of these 9 reputation agencies.  The Twitter
data was collected on November 17$^\text{th}$ 2014 at 9:00 pm for a duration of 24 hours. The data was obtained using the {\em Twitter Streaming API}\footnote{\url{https://dev.twitter.com/streaming/overview}} and a custom python script. The sentiment of the tweets and retweets is computed using {\em TextBlob}\footnote{TextBlob - Python based Text Processing Tool Suit, \url{http://textblob.readthedocs.org/en/dev/}}. The social network contains 10,656 nodes with 11,054 edges.
 As   illustrated in Fig.\ref{fig:reviewnet},  numerous  nodes retweet based on tweets posted by these 9 agencies. The edge intensity in Fig.\ref{fig:reviewnet} can be used to gain intuition on the dynamics of retweets. For the nodes with high {\it in degree} such as @IGN and @HarvardBiz the retweets typically occur in a short period of time on the order of 1-12 hours. This behavior has been observed in popularity dynamics of papers and youtube videos~\cite{SWSB14} and is associated with a decrease in the ability to attract new attention after ageing. 
 

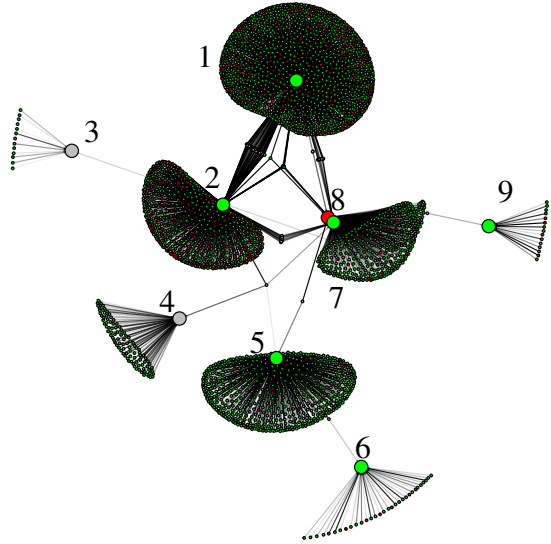
\begin{figure}[h]
	\begin{tikzpicture}[scale =1,transform shape, american voltages]
\node[inner sep=0pt] at (-0.5,0) {\includegraphics[angle=180,scale=0.6,trim = 30mm 30mm 20mm 20mm, clip]{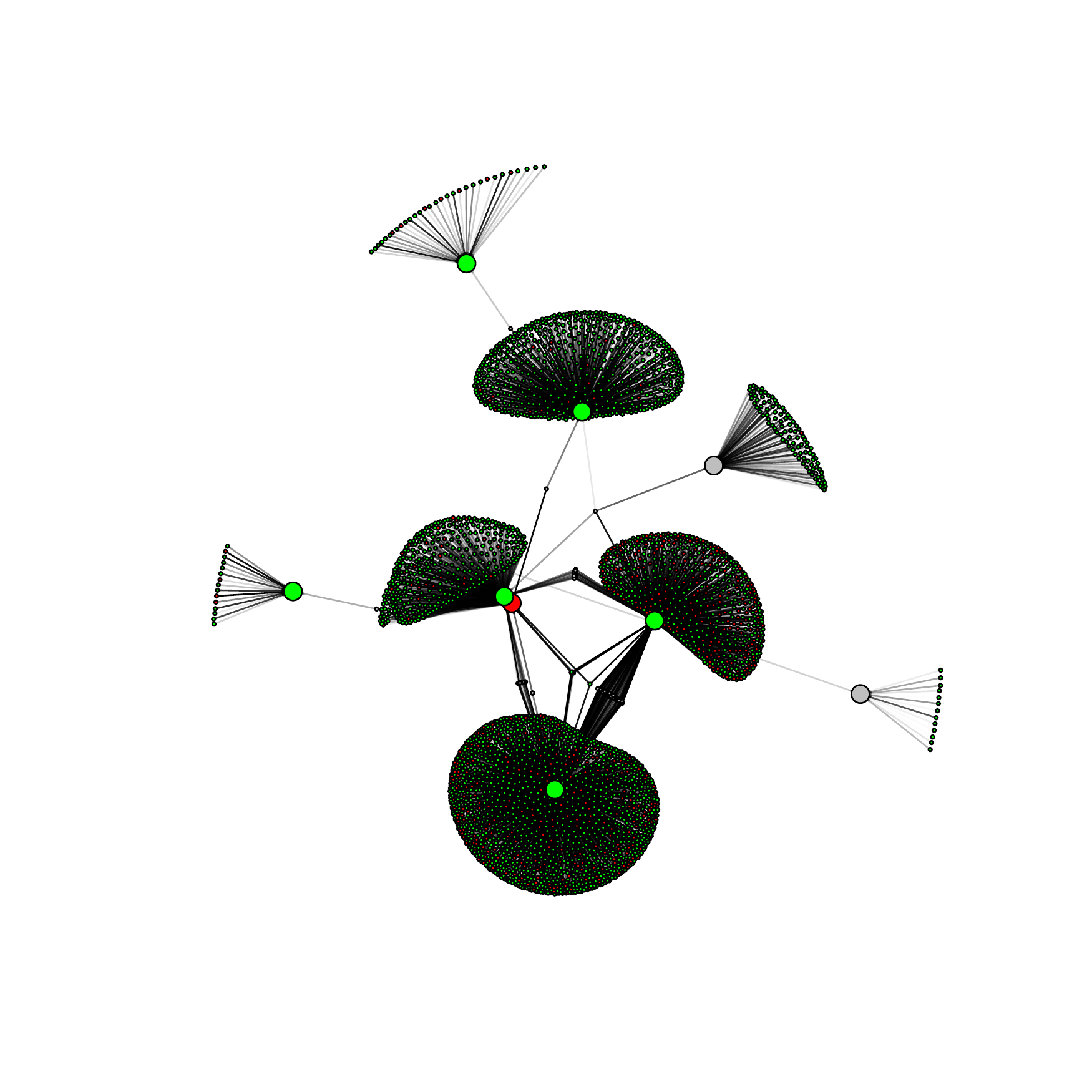}};
\node at (-1.5,3){\large 1};
\node at (-1.4,1.4){\large 2};
\node at (-3,2){\large 3};
\node at (-2,-0.25){\large 4};
\node at (-0.8,-0.8){\large 5};
\node at (0.6,-2.2){\large 6};
\node at (0.25,-0.2){\large 7};
\node at (0.27,1.1){\large 8};
\node at (2.5,1.2){\large 9};
\end{tikzpicture}
\vspace{-10pt}
	\caption{Snapshot of the estimated retweet network obtained by tracking real-time tweets of reputation agencies. The labels $1,2,\dots,9$ correspond to the Twitter accounts @IGN, @gamespot, @AmznMovieRevws, @creativereview, @HarvardBiz, @techreview, @pcgamer, @RottenTomatoes, @LAReviewofBooks. The data is collected over a period of 24 hours starting from November 17$^\text{th}$ 2014 at 9:00 pm. 
	The reputation agency nodes are denoted by large circles and retweeting nodes 
	(followers) by small circles.
	The sentiment of the tweet published by the reputation
	agency is denoted by color: red is negative, green positive, and gray is neutral.  The time of the retweet is indicated by the shade intensity of the edges
	of the graph:  the lighter the shade of an edge, the later the retweet was posted.}
	\label{fig:reviewnet}
\end{figure}

In the following analysis, the aim is to determine if the Twitter followers of an online reputation agency exhibit  utility maximization behavior 
in response to the tweet published by the reputation agency. 
The index  $\tindx=1,2,\ldots,$ denotes  epochs at which the reputation agency publishes its tweets.
 \footnote{The  average time interval between tweets for @IGN, @gamespot, and @HarvardBiz are respectively, 31, 33, and 34 minutes.}
To apply Afriat's theorem (\ref{eqn:AfriatFeasibilityTest}), the public belief is defined by $$\probe_\tindx^i=[{\tt \# followers}, {\tt neutrality}]$$ for each
 reputation agency
$i$. The ${\tt \#followers}$  is the  number of followers of a tweet published by the online reputation agency. The ${\tt neutrality}$ of the tweet published by the 
reputation agency is computed as 
$1/ | {\tt polarity}|$ where the ${\tt polarity}$ of a tweet is computed using TextBlob.
The associated response taken by Twitter users that retweet in the network is given by 
$$\response_\tindx^i=[\Delta t, {\tt \#retweets}]. $$ $\Delta t$ denotes the time between the tweet (published by the agency) and the first retweet (of a follower). ${\tt \#retweets}$ denotes the total number of retweets generated by followers prior to the next tweet from the  reputation agency~$i$.
 Next we need to justify the linear budget in (\ref{eqn:singlemaximization})  to use Afriat's theorem. It is clear that
 as the number of followers of a reputation agency increases, the number of retweets will increase. Consequently one expects   a decrease in time between the first retweet $\response_\tindx(1)$ as the number of followers $\probe_\tindx(1)$  increases. The results in \cite{SD13} suggest that the higher the polarity of a tweets published by the reputation agency, the larger the number of retweets. So  we expect that as the {\it neutrality} (i.e. the lower the polarity) of the tweet increases the resulting number of retweets $\response_\tindx(2)$ will decrease. So it is reasonable to assume the existence of a social impact budget $I_\tindx$ for the utility maximization test (\ref{eqn:AfriatFeasibilityTest}) which satisfies $I_\tindx=\probe_\tindx^\prime\response_\tindx$. We construct the datasets $\dataset_i$ for each reputation agency $i\in\{1,2,\dots,9\}$ from the Twitter data collected on November 17$^\text{th}$ 2014 at 9:00 pm for a duration of 24 hours. The dataset $\dataset_i=\{(\probe_\tindx^i,\response_\tindx^i): t\in\{1,2,\dots,T^i\}\}$ was constructed using the public belief $\probe_\tindx^i$, response $\response_\tindx^i$, and total number of tweets $T^i$ for the reputation agency $i$. Note that $t\in\{1,\dots,T^i\}$ denotes the tweets published by the reputation agency.


\subsubsection{Results} We found that each of the Twitter datasets $\dataset_i$ satisfy the utility maximization test (\ref{eqn:AfriatFeasibilityTest}). Using (\ref{eqn:estutility}) from Afriat's Theorem the associated utility function for the reputation agencies @IGN, @gamespot, @RottenTomatoes, and @LAReviewBooks is provided in Fig.\ref{fig:retweeta}-\ref{fig:retweetd}. The utility function for the other 5 agencies are omitted as only minor differences are present compared with the utility functions provided in Fig.\ref{fig:reviewnetutility}.  Several key observations can be made from the results:
\begin{compactenum}[(a)]
\item Given that $\dataset_i$ satisfies (\ref{eqn:AfriatFeasibilityTest}), this suggests that the number of followers $\probe^i_\tindx(1)$ 
and neutrality $\probe^i_\tindx(2)$ of the tweet contributes to the retweet dynamics (i.e. the delay before the first tweet $\response^i_\tindx(1)$ and the 
total number of retweets $\response^i_\tindx(2)$). 

\item The utility functions provided in Fig.\ref{fig:reviewnetutility} suggest that Twitter users prefer to increase the delay of retweeting $\response^i_\tindx(1)$ compared with increasing the total number of retweets $\response^i_\tindx(2)$. The results also suggest that as the delay before the first retweet increases the associated number of retweets decreases. This effect is pronounced in Fig.\ref{fig:retweetd} where the first and only retweet occurs approximately 2000 seconds after the original tweet.
\item The revealed preferences of the reputation agencies, represented by the utility functions in Fig.\ref{fig:retweeta}-\ref{fig:retweetd}, are not identical. 
\end{compactenum}
Observation (a) is straightforward as a change in the number of followers will affect the time of a retweet. Additionally, as suggested in \cite{SD13}, as neutrality of the tweet increases the associated total number of retweets is expected to decrease. Observation (b) illustrates an interesting characteristic of how users retweet to reputation agencies, it suggests that Twitter users prefer to increase the time prior to retweeting compared to increasing the total number of retweets. This result is caused by social features of users which include the content of the tweet, and the external context during which the tweet is posted~\cite{SIDB12}. In \cite{SIDB12} over 250 million tweets are collected and analyzed and it was found that a high number of followers does not necessarily lead to shorter retweet times. A possible mechanism for this effect is that tweets that are exposed to a larger number of Twitter users then an individual user is less likely to retweet--this effect is known as {\it diffusion of responsibility}. For large retweet times $\response^i_\tindx(1)$ we observe that the total number of retweets is significantly reduced compared to short retweet times, refer to Fig.\ref{fig:retweeta} and Fig.\ref{fig:retweetd}. This result has been observed in popularity dynamics~\cite{SWSB14} and is associated with an ageing effect--that is, the interest of the tweet decreases with time. Observation (c) is expected as different reputation agencies are expected to have different users retweeting. To quantify this result for the constructed ordinal (i.e. identical for any monotonic transformation) utility functions the comparison of preferences is achieved using the {\it marginal rate of substitution} defined by:
\begin{equation}
\text{MRS}_{12} = \frac{\partial{u}/\partial{\response^i(1)}}{\partial{u}/\partial{\response^i(2)}}. 
\label{eqn:mrs}
\end{equation}
In (\ref{eqn:mrs}), $\text{MRS}_{12}$ defines the amount of $\response^i(2)$ the Twitter users are willing to give up for 1 additional unit of $\response^i(1)$. From Fig.\ref{fig:retweeta}-\ref{fig:retweetd} it is clear that $\text{MRS}_{12}>1$ suggesting users prefer $\response^i(1)$ compared with $\response^i(2)$. Additionally notice that the $\text{MRS}_{12}$ for each of the reputation agencies in Fig.\ref{fig:retweeta}-\ref{fig:retweetd} illustrating that the associated behavior of each is characteristically different. 


\begin{figure}[h]
\subfigure[@IGN Twitter account.]{\label{fig:retweeta}
	\begin{tikzpicture}[scale =1,transform shape, american voltages]
\node[inner sep=0pt] at (0,0) {\includegraphics[angle=0,scale=0.045,trim = 80mm 30mm 180mm 30mm, clip]{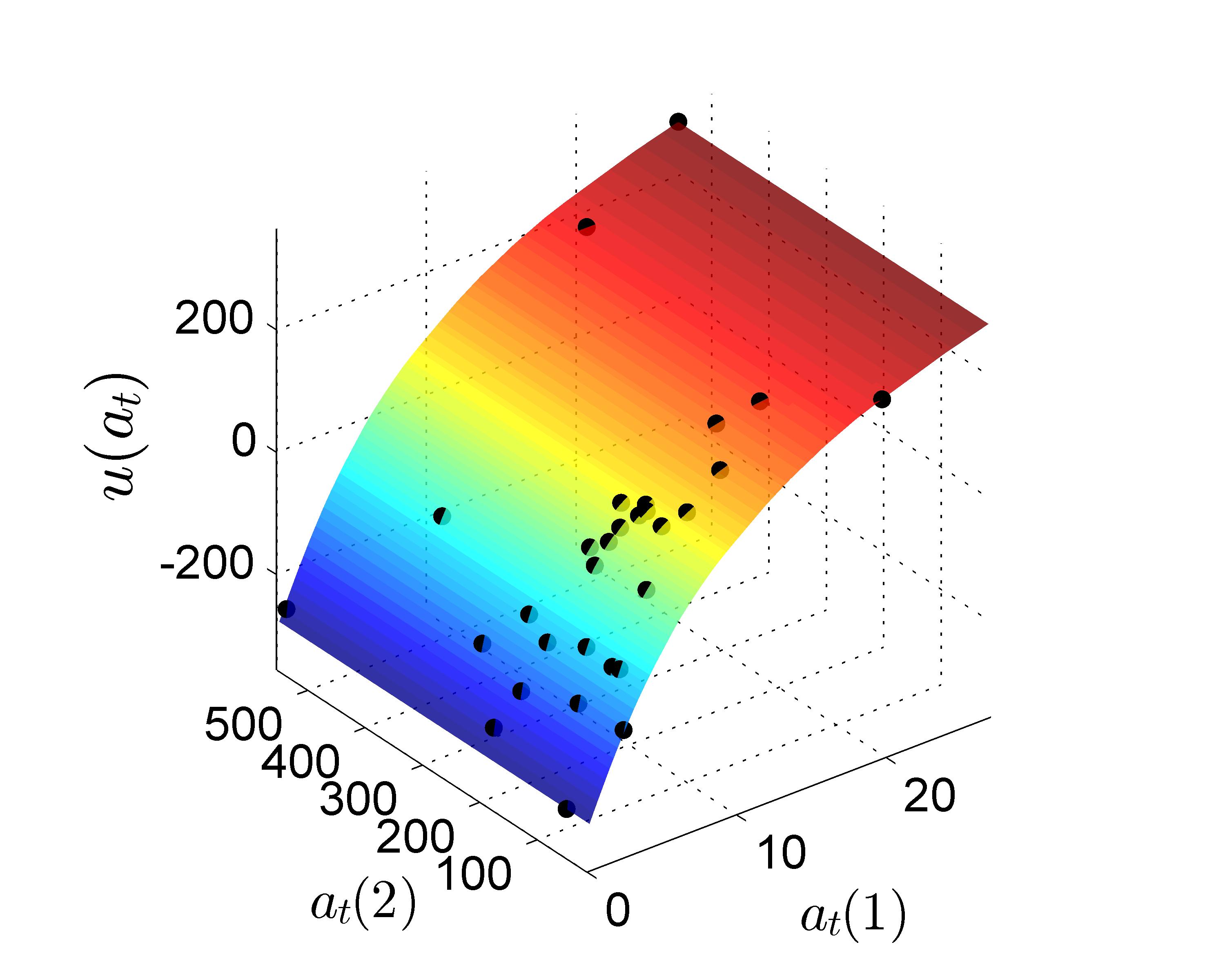}};
\end{tikzpicture}}
\subfigure[@gamespot Twitter account.]{\label{fig:retweetb}
	\begin{tikzpicture}[scale =1,transform shape, american voltages]
\node[inner sep=0pt] at (0,0) {\includegraphics[angle=0,scale=0.045,trim = 80mm 30mm 110mm 30mm, clip]{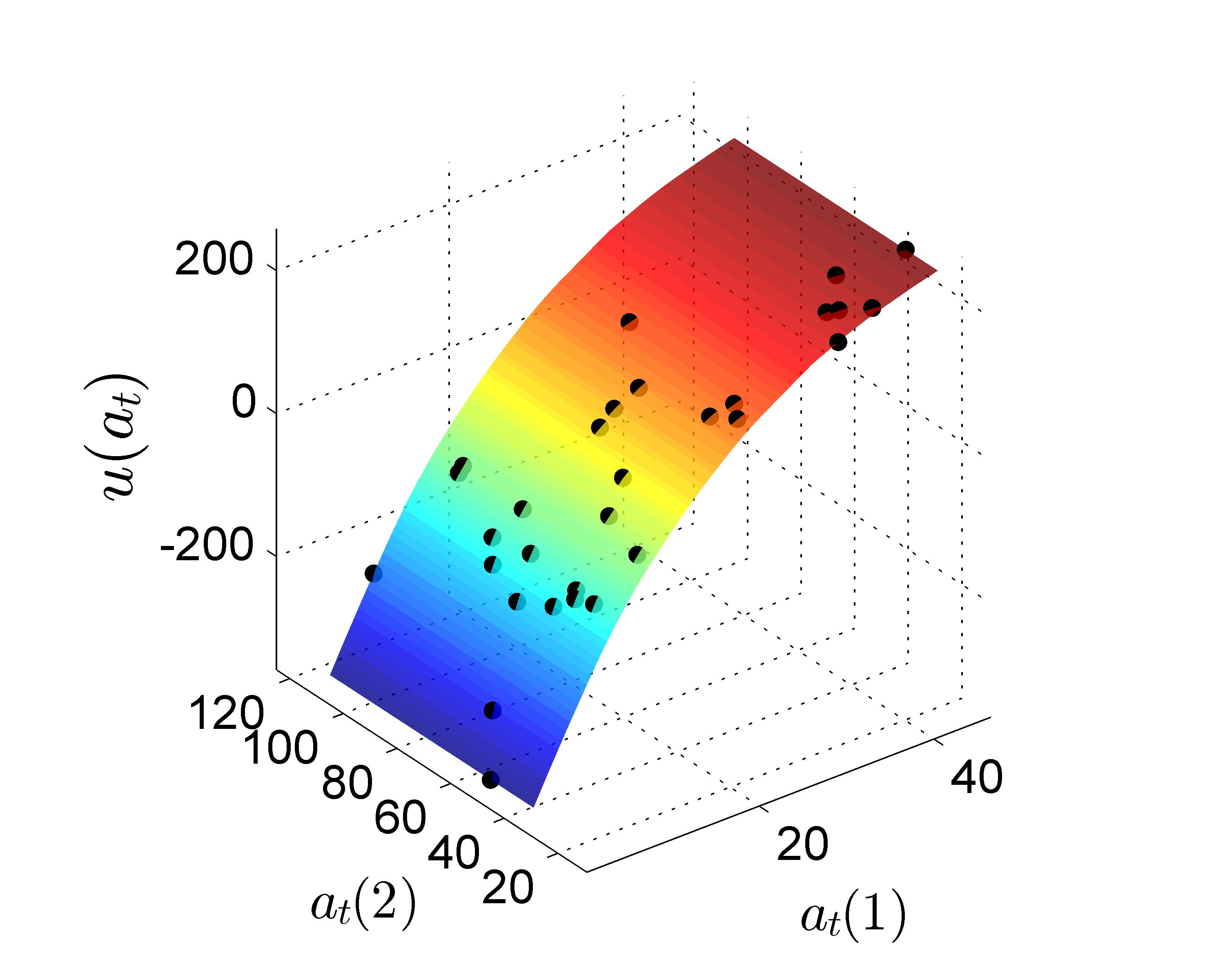}};
\end{tikzpicture}} \\
\subfigure[@HarvardBiz Twitter account.]{\label{fig:retweetc}
	\begin{tikzpicture}[scale =1,transform shape, american voltages]
\node[inner sep=0pt] at (0,0) {\includegraphics[angle=0,scale=0.045,trim = 70mm 30mm 120mm 30mm, clip]{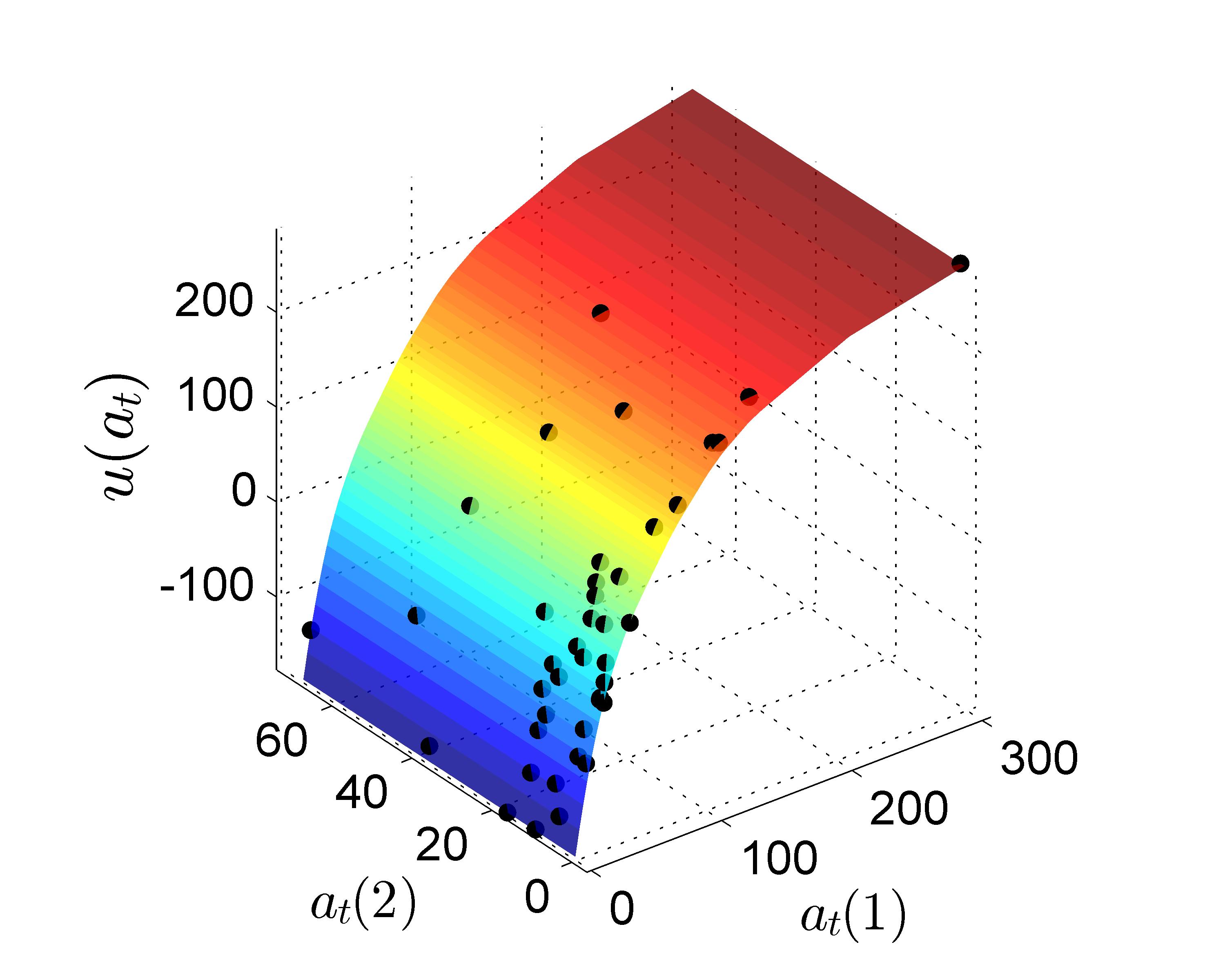}};
\end{tikzpicture}}
\subfigure[@LAReviewofBooks Twitter account.]{\label{fig:retweetd}
	\begin{tikzpicture}[scale =1,transform shape, american voltages]
\node[inner sep=0pt] at (0,0) {\includegraphics[angle=0,scale=0.045,trim = 70mm 30mm 150mm 30mm, clip]{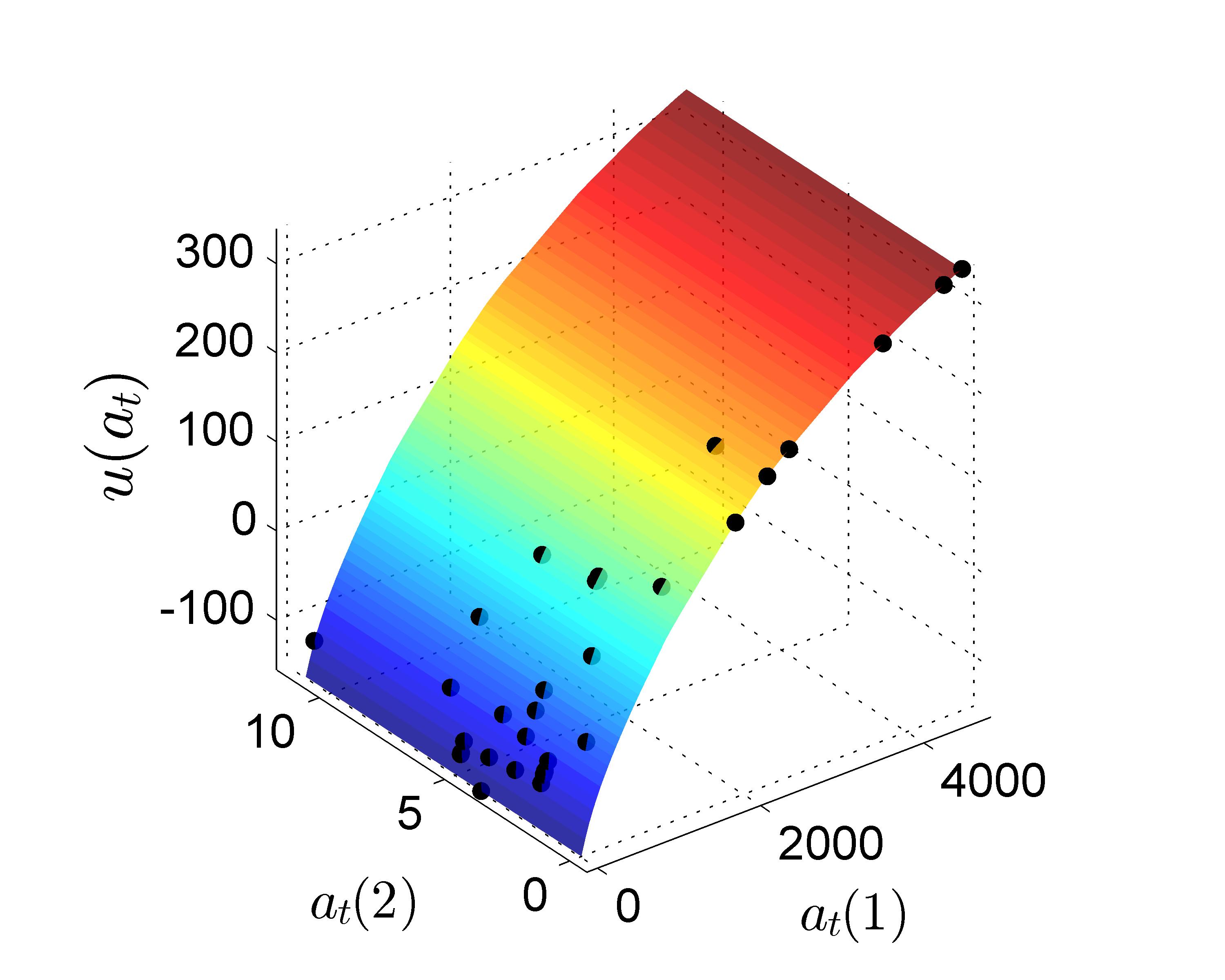}};
\end{tikzpicture}}
	\caption{Estimated utility function $\utility(\response_\tindx)$ using the Twitter datasets $\dataset_i$ for $i\in\{1,2,5,9\}$  defined in Sec.\ref{sec:twitterreview}  constructed using the  (\ref{eqn:estutility}) from Afriat's Theorem. 
	Note that $\response_\tindx(1)$ is the number of retweets, and $\response_\tindx(2)$ has units of seconds.}
	\label{fig:reviewnetutility}
\end{figure}

\subsubsection{Social Learning}
To interpret the above results in terms of social learning, we next show that the response (action)  $\response_\tindx$ and public belief
$\probe_\tindx$ at epoch $\tindx$ determine the public belief
 $\probe_{\tindx+1}$ at epoch $\tindx+ 1$.
We found that the following auto-regressive time series model
\begin{equation}
\probe_{\tindx+1}(1) = \probe_{\tindx}(1)+b\, \response_{\tindx}(2)+\noise_\tindx
\label{eqn:arxf}
\end{equation}
driven by a zero mean noise process $\{\noise_\tindx\}$
yields an excellent fit, where 
the  parameter $b$ is  estimated from the Twitter data  using least-squares. Note that $\probe_{\tindx+1}(1)-\probe_{\tindx}(1)$ in (\ref{eqn:arxf}) models the total number of new followers resulting from the the total number of retweets $\response_{\tindx}(2)$. To test the accuracy of (\ref{eqn:arxf}), we selected the reputation agencies @IGN, @gamespot, and @HarvardBiz. The experimentally measured and numerically predicted results are displayed in Fig.\ref{fig:arxfollow}. As seen (\ref{eqn:arxf}) accurately predicts the public opinion based on the response of the Twitter users. The {\it mean absolute percentage error} using the AR model for the reputation agencies @IGN, @gamespot, and @HarvardBiz are,
respectively,  0.49\%, 0.67\%, and 0.41\%. 

\begin{figure}[h]
\center
	\includegraphics[angle=0,width=3.4in,,trim = 10mm 0mm 0mm 00mm, clip]{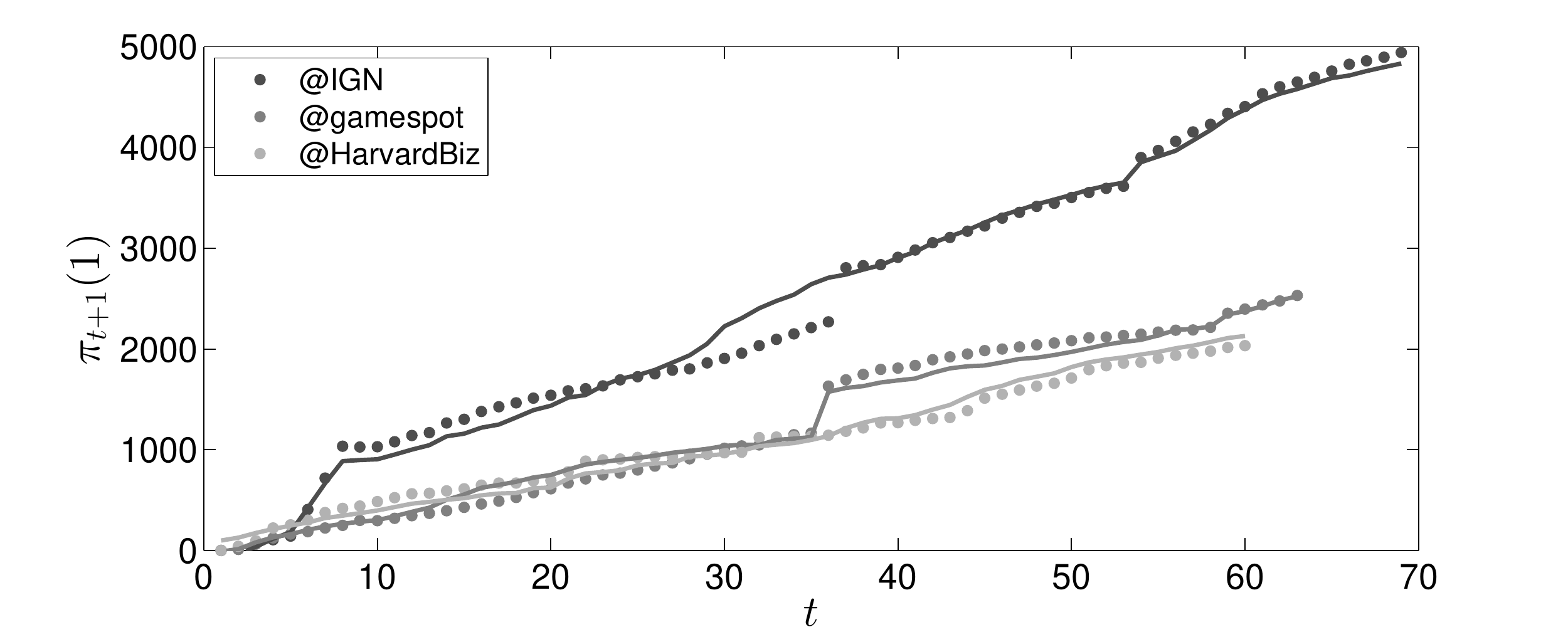}
\vspace{-15pt}
	\caption{Accuracy of the AR model (\ref{eqn:arxf}) for the social learning dynamics of  public belief $\probe_{\tindx+1}(1)$ of reputation agencies: @IGN, @gamespot, and @HarvardBiz. The dots represent the actual number of Twitter followers while the solid line indicates the predicted number of followers using
	the AR model. The initial number of followers at $t=0$ is set to zero in the plot. The estimated parameter values of $b$ in
	(\ref{eqn:arxf}) 
	for agencies @IGN, @gamespot, and @HarvardBiz are respectively: $0.4358$, $0.8132$, and $0.3825$. The average time interval between tweet epochs $t$ for @IGN, @gamespot, and @HarvardBiz are respectively: 31, 33, and 34 minutes.}
	\label{fig:arxfollow}
\end{figure}

\section{Numerical Examples} \label{sec:numerical}

This section illustrates the incest removal  algorithm of social learning Protocol  \ref{protocol:socialcons}  for two different types of social networks.  Then incest removal
in expectation polling is illustrated for these two networks.

\subsection{Social Learning and Incest  in Corporate Network}

 Consider   the network of Fig.\ref{fig:corporate} which depicts a corporate social network. 
Nodes 1 and 10 denote the same  senior level manager. Nodes 2 and 8 denote a mid-level manager; and nodes 3 and 9 denote another mid-level manager.
The two mid-level managers  attend a presentation (and therefore opinion) by a senior manager
to determine (estimate) a specific parameter $x$ about the company. 
Each mid-level  manager then  makes recommendations (decisions) and convey these to two of their workers.
 These workers eventually report back to their mid-level managers who in turn report back to the senior manager.  
 The edge $(1,10)$ indicates that the senior manager recalls her decision at node 1 when making her decision at node 10.
Similarly for edges  (2,8)  and (3,9).

 \begin{figure}[h]
\centering
\begin{tikzpicture}[->,>=stealth',shorten >=1pt,auto,node distance=2cm,
  thick,main node/.style={circle,fill=white!12,draw,font=\sffamily},scale=0.55, every node/.style={transform shape}]
 
  \node[main node] (1) {1};
  \node[main node] (2) at (2,2) {2};
  \node[main node] (3) at (2,-2){3};
  \node[main node] (4) at (3,3) {4};
  \node[main node] (5) at (3,1) {5};
  \node[main node] (6) at (3,-1)  {6};
  \node[main node] (7) at (3,-3)  {7};
  \node[main node] (8) at (4,2) {8};
  \node[main node] (9) at (4,-2) {9};
  \node[main node] (10) at (6,0) {10};
  \path[every node/.style={font=\sffamily\small}]
    (1) edge node {} (2)
        edge node {} (3)
        edge node {} (10)
    (2) edge node {} (4)
        edge node {} (8)
        edge node {} (5)
    (3) edge node  {} (7)
        edge node  {} (6)
        edge node {} (9)
    (4) edge node  {} (8)
    (5) edge node {} (8)
    (6) edge node {} (9)
    (7) edge node {} (9)
    (8) edge node {} (10)
    (9) edge node {} (10);
\end{tikzpicture}
\caption{Social Learning in Corporate  Network comprising of higher level manager, mid level manager and workers.
Data incest arises at nodes 8, 9 and 10.}
\label{fig:corporate}
\end{figure}
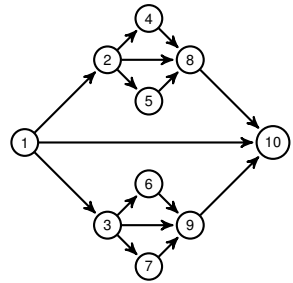
We ran 1000 independent  simulations
with the following parameters in the social learning  Protocol \ref{protocol:socialcons}: 
  $\X=\{1,2,\ldots,X\}$, $\A=\{1,2,\ldots,A\}$, $\Y=\{1,2,\ldots,Y\}$, $X=10$, $A=10$, $Y=20$, true state $x $ uniform on $\X$, prior $\belief_0$  uniform on $\X$,  observation probabilities
$B_{x,y}  \propto  \exp- \frac{(y-x)^2}{2}$,
and costs
$c(i,a) = | \frac{A}{X} x - a| $.

 It is clear from Fig.\ref{fig:corporate} that incest arises at nodes 8, 9 and  10 due to  multiple information paths.
Table \ref{tab:corporate} displays the mean square error (MSE) of the state estimates computed at these nodes.
The incestious estimates were computed using naive data fusion (\ref{eq:dataincest}). The incest free estimates were computed using the incest removal algorithm of Theorem \ref{thm:socialincestfilter}.
It is verified  that the achievability condition of Theorem \ref{thm:sufficient} holds;  hence incest can be removed completely given the information from single hop nodes
(immediate friends).
Table \ref{tab:corporate}  shows that incest removal results in substantially more accurate estimates, particularly at node 10 which is the most informed node (since all nodes eventually
report to node 10).

\begin{table}
\begin{subfigure} \centering
\begin{tabular}{|L|c|c|c|} \hline 
Node  &  8 & 9 & 10 \\ \hline
MSE with incest  & 0.3666 & 0.3520 & 0.3119 \\ \hline
MSE with incest removal & 0.2782 & 0.2652& 0.1376\\ \hline
\end{tabular}
\caption{Corporate Network of Fig.\ref{fig:corporate}} \label{tab:corporate}
\end{subfigure}

\begin{subfigure} \centering
\begin{tabular}{|L|c|c|c|c|c|} \hline 
Node  &  5 & 6  & 7 & 8  & 9  \\ \hline
MSE with incest  & 0.3246  & 0.3420 & 0.3312  & 0.3149 & 0.3134 \\ \hline
MSE with incest removal & 0.2799 & 0.2542 & 0.2404  & 0.2267 & 0.2200\\ \hline
\end{tabular}
\caption{Mesh Network of Fig.\ref{fig:mesh}} \label{tab:mesh}
\end{subfigure}

\caption{Effect of incest removal in social learning. The mean square errors (MSE) in the state estimate  were obtained by averaging over 1000 independent simulations.}
\end{table}

\subsection{Social Learning and Incest in Mesh Network}
Consider the mesh network depicted in Fig.\ref{fig:mesh} with  parameters identical to the previous example. Incest arises at nodes 5,6 8 and 9.
Also incest propagates to node 6 from node 5.
It is verified that the condition of Theorem \ref{thm:sufficient} holds for the mesh network and so incest can be completely removed.
Table \ref{tab:mesh} displays the mean square error of naive data fusion and incest removal. and shows that  substantial improvements occur in the state estimate with
incest removal.

\begin{figure}
\centering
\begin{tikzpicture}[->,>=stealth',shorten >=1pt,auto,node distance=2cm,
  thick,main node/.style={circle,fill=white!12,draw,font=\sffamily},scale=0.8, every node/.style={transform shape}]
    \node[main node] (1) at (0,0) {1};
  \node[main node] (2) at (0,-1) {2};
  \node[main node] (3) at (0,-2){3};
  \node[main node] (4) at (1,-2) {4};
  \node[main node] (5) at (1,-1) {5};
  \node[main node] (6) at (1,0)  {6};
  \node[main node] (7) at (2,0)  {7};
  \node[main node] (8) at (2,-1) {8};
  \node[main node] (9) at (2,-2) {9};
   \path[every node/.style={font=\sffamily\small}]
    (1) edge node {} (2)
        edge node {} (6)
         (2) edge node {} (3)
        edge node {} (5)
        (3) edge node  {} (4) 
         (4) edge node {} (5)
        edge node {} (9)
 (5) edge node {} (6)
        edge node {} (8)
          (6) edge node  {} (7) 
             (7) edge node  {} (8) 
   (8) edge node  {} (9) ;
\end{tikzpicture}
\caption{Social Learning  in  Mesh Network. Data Incest occurs at nodes 5, 6, 8 and 9.} \label{fig:mesh}
\end{figure}
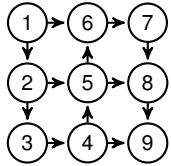

\subsection{Expectation Polling using Incestious Belief}
We illustrate the results of Sec.\ref{sec:vote} for incest removal in  expectation polling.
Consider  the network  of Fig.\ref{fig:corporate} with   the edge  $(1,10 )$  omitted. Node 10 denotes the pollster.
 Assume  the pollster cannot sample  node 1 
 so that exact data incest removal is impossible. (It can be verified that the necessary and sufficient condition (\ref{constraintnetwork})
of Theorem \ref{thm:sufficient}  does not hold if node 1 is omitted.) 
Given the incestious  beliefs (expectations) $\belief_m$ from the sampled voters $m \in \Req$, we use the estimator of Theorem~ \ref{thm:estpollster}  to compute the 
 conditional mean estimate $\hat{x} =  \E\{x|\belief_m ,m \in \Req\}$ of the leading candidate. 

For the simulation we chose $\X=\{1,2\}$ (two candidates), prior $[0.4,\; 0.6]^\p$, $\Y=\{1,2\}$, $B = \begin{bmatrix} 0.8 & 0.2 \\ 0.2 & 0.8 \end{bmatrix}$.
Table \ref{tab:expcorp} displays  the
mean square error of  $\hat{x}$  for different
recruited sample sets $\Req$.  The table also displays the MSE when all nodes are sampled in which case optimal incest removal is achieved.
These are   compared with the naive incestious estimator using  $\belief_{10}$ computed via Protocol \ref{polling}.

\begin{table}
\begin{subfigure} \centering
\begin{tabular}{|c|c|} \hline 
sampled voters $\Req$  & MSE     \\ \hline
$\{8,9,10\}$ &   0.0250 \\
$\{4,\ldots,10\}$ & 0.0230 \\
$\{2,\ldots,10\}$ & 0.0161 \\
$\{1,\ldots,10\}$ & 0.0118 \\ \hline
Naive (incest) & 0.0513 \\ \hline
\end{tabular}
\caption{Corporate Network of Fig.\ref{fig:corporate} with edge (1,10) omitted.} \label{tab:expcorp}
\end{subfigure}

\begin{subfigure} \centering
\begin{tabular}{|c|c|} \hline 
sampled voters $\Req$  & MSE   \\ \hline
$\{4,8,9\} $ &  0,0314 \\
$\{2,4,8,9\}$ &  0.0178  \\
$\{1,\ldots,9\}$ &  0.0170\\ \hline
Naive (incest) & 0.0431 
\\ \hline
\end{tabular}
\caption{Mesh Network of Fig.\ref{fig:mesh}} \label{tab:expmesh}
\end{subfigure}

\caption{Effect of incest removal in expectation polling using  the estimator (\ref{eq:pollest}). The mean square errors (MSE) in the state estimate  were obtained by averaging over 1000 independent simulations.}
\end{table}

Finally, consider expectation polling of voters in the  mesh network of Fig.\ref{fig:mesh}. Table \ref{tab:expmesh} displays the mean square errors of the conditional mean 
estimates for different sampled nodes and the naive incestious estimate.

\section{Conclusions and Extensions}
This paper considered data incest in reputation and expectation polling systems. In reputation systems, data incest arose in a multi-agent social learning model. A necessary and sufficient
condition on the adjacency matrix of the directed graph was given for exact incest removal at each stage. For expectation polling systems, it was shown that even if incest  propagates in the network, the posterior of the state can be estimated based on the incestious beliefs.
Finally by analyzing Twitter data sets associated with several  online reputation agencies we used Afriat's theorem of revealed preferences
to show utility maximization behavior and social learning.
In future work, it is  worthwhile  extending the framework in this paper to active sensing and sequential decision making. For the case of classical social learning,  
\cite{Kri12} deals with sequential quickest detection and stochastic control. Extending this to data incest information patterns is  challenging and non-trivial.

\subsubsection*{Acknowledgement} The experimental data in Sec.\ref{sec:expt} was obtained from Dr.\  Grayden Solman and Prof.\ Alan Kingstone of the Department of Psychology,
University of British Columbia. The experiments were conducted by them in the fall semester of 2013. The  data analysis of Sec.\ref{sec:expt} and Fig.\ref{Fig:Samplepath} was prepared  by Mr.\ Maziyar Hamdi. A complete description  of our results in the psychology experiment is in the preprint \cite{HSK14}.

\appendix
\subsection{Illustrative Example}
We provide here an example of the data incest problem setup of Sec.\ref{sec:classicalsocial}.
 Consider  $S=2 $ two agents with  information flow graph for three  time points $k=1,2,3$ depicted in Fig.\ref{sample}  characterized by 
the family of DAGs  $\{G_1,\ldots,G_7\}$. 
The adjacency matrices  
$A_1,\ldots,A_7$ are constructed as follows:  $A_n$ is the upper left  $n\times n$ submatrix of $A_{n+1}$ and 
$${\small A_7 = \begin{bmatrix}
0 & 0 & 1 & 1 & 0 & 0 & 1 \\
0 & 0 & 0 & 1 & 0 & 0 & 0 \\
0 & 0 & 0 & 0 & 1 & 0 & 0 \\
0 & 0 & 0 & 0 & 0 & 1 & 0 \\
0 & 0 & 0 & 0 & 0 & 0 & 1 \\
0 & 0 & 0 & 0 & 0 & 0 & 1 \\
0 & 0 & 0 & 0 & 0 & 0 & 0 \\
\end{bmatrix}} .$$

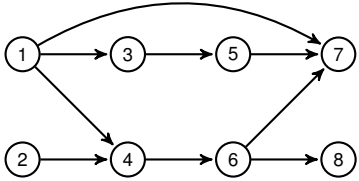
\begin{figure}[h]
\centering
\begin{tikzpicture}[->,>=stealth',shorten >=1pt,auto,node distance=2cm,
  thick,main node/.style={circle,fill=white!12,draw,font=\sffamily},scale=0.7, every node/.style={transform shape}]
 
  \node[main node] (1) {1};
  \node[main node] (2) [below  of=1] {2};
  \node[main node] (3) [right of=1] {3};
  \node[main node] (4) [below  of=3] {4};
  \node[main node] (5) [right of=3] {5};
  \node[main node] (6) [right of=4] {6};
  \node[main node] (7) [right of=5] {7};
  \node[main node] (8) [right of=6] {8};

  \path[every node/.style={font=\sffamily\small}]
    (1) edge node  {} (4)
          edge [bend left] node[left] {} (7)     
         edge node {} (3)
    (2) edge node {} (4)
    (3) edge node {} (5)
    (4) edge node {} (6)
    (5) edge node {} (7)
    (6) edge node {} (7)
         edge node {} (8);
\end{tikzpicture}

\caption{Example of  information flow network with $S=2$  two agents, namely   $s\in \{1,2\}$ and  time points $k=1,2,3,4$.  Circles represent the nodes indexed by $n= s + S(k-1)$
in the social network and each edge depicts a communication link between two nodes.}
\label{sample}
\end{figure}

Let us explain these matrices.
Since nodes 1 and 2 do not communicate,
clearly $A_1$ and $A_2$ are zero matrices. Nodes 1 and  3 communicate, hence  $A_3$ has a single
one, etc. Note that if nodes 1,3,4 and 7 are assumed to be the same individual, then at node 7, the individual remembers what happened at node 5 and node 1, but not node 3.
This models the case where the individual has selective memory and remembers certain highlights; we discuss this further in Sec.\ref{sec:discussion}.
From (\ref{eq:history}) and (\ref{eq:full}),  
$$\history_7=\{1,5,6\}, \quad \full_7 = \{1,2,3,4,5,6\} $$
where $\history_7$ denotes all one hop links to node 7 while $\full_7$ denotes all multihop links to node 7.

Using (\ref{eq:tc}), the transitive closure matrices $T_1,\ldots,T_7$   are given by: $T_n$ is the upper left  $n\times n$ submatrix of $T_{n+1}$ and 
$$ {\small T_7 = \begin{bmatrix}
1 & 0 & 1 & 1 & 1 & 1 & 1 \\
0 & 1 & 0 & 1 & 0 & 1 & 1 \\
0 & 0 & 1 & 0 & 1 & 0 & 1 \\
0 & 0& 0 & 1 & 0 & 1 & 1 \\
0 & 0 & 0 & 0 & 1 & 0 & 1 \\
0 & 0 &0 & 0 & 0 & 1 & 1 \\
0 & 0 &0 & 0 & 0 & 0 & 1 
\end{bmatrix}.}$$
Note that $T_n(i,j) $ is non-zero only for $i\geq j$ due to causality since information sent by an agent can only arrive at another social group at a later time instant. The weight vectors are then
obtained from  (\ref{eq:socialconstraintestimate}) as
\begin{align*}
w_2 &= \begin{bmatrix}0\end{bmatrix},\quad
w_3 = \begin{bmatrix}1 & 0 \end{bmatrix}^\p,\quad
w_4 = \begin{bmatrix}1 & 1 & 0\end{bmatrix}^\p ,\\
w_5 & = \begin{bmatrix} 0  & 0 & 1 & 0\end{bmatrix}^\p, \quad
w_6 =  \begin{bmatrix} 0  & 0 & 0 & 1 & 0\end{bmatrix}^\p, \\
w_7 &= \begin{bmatrix} -1  & 0 & 0 & 0 & 1 & 1\end{bmatrix}^\p.
\end{align*}
 $w_2$ means that node $2$ does not use the estimate from node $1$. This formula is consistent with the  constraints on  information flow because 
the estimate from node $1$ is not available to node $2$; see Fig.\ref{sample}.
$w_3$ means that node $3$ uses estimates from nodes $1$; $w_4$ means
that  node $4$  uses estimates only from node $1$ and node $2$.  As shown in Fig.\ref{sample}, the mis-information propagation occurs at node $7$ since there are multiple
paths from node 1 to node 7. The vector $w_7$ says that node 7 adds estimates from nodes $5$ and $6$ and removes estimates from node $1$ to avoid triple counting of these estimates already integrated into estimates from nodes $3$ and $4$. Using the algorithm (\ref{eq:socialconstraintestimate}), incest is completely prevented in this example. 

Here is
an example in which exact incest removal is impossible.
Consider the information flow graph of  Fig.\ref{sample} but with the edge between node 1 and node 7  deleted. Then $A_7(1,7) = 0$ while $w_7(1) \neq 0$, and therefore the condition (\ref{constraintnetwork}) does not hold. Hence exact incest  removal is not possible for this case. In Sec.\ref{sec:vote} we compute the Bayesian  estimate of the underlying state when incest cannot be removed.

\subsection*{Proof of Theorem \ref{thm:monotone}}
The proof uses MLR stochastic dominance (defined in footnote \ref{footnotemlr}) and the following  single crossing
condition:

\begin{definition}[Single Crossing  \cite{Ami05}] \label{def:scc}
$g:\Y\times \A\rightarrow \reals $ satisfies a single crossing condition in  $(y,a)$  if $g(y,a) - g(y,\bar{a}) \geq 0$ implies $g(\bar{y},a)- g(\bar{y},\bar{a}) \geq 0$ for $\bar{a}>a$ and $\bar{y} > y$.
Then $a^*(y) = \argmin_{a} g(y,a)$ is increasing in $y$. \qed
\end{definition}

By (A1) it is  verified that the Bayesian update satisfies
$$ \frac{B_y \belief}{\ones^\p B_y \belief}  \lr  \frac{B_{y+1} \belief}{\ones^\p B_{y+1} \belief} $$
where $\lr$ is the MLR stochastic order. (Indeed, the MLR order is closed under conditional expectation and this is the reason why it is widely used in Bayesian analysis.)
By submodular assumption (A2), 
$c_{a+1} - c_a $ is a vector with decreasing elements.
Therefore 
$$(c_{a+1} - c_a )^\p \frac{B_y \belief}{\ones^\p B_y \belief}  \geq (c_{a+1} - c_a )^\p\frac{B_{y+1} \belief}{\ones^\p B_{y+1} \belief} $$
Since the denominator is non-negative, it follows that 
$ (c_{a+1} - c_a )^\p B_{y+1} \belief \geq 0 \implies (c_{a+1} - c_a )^\p B_{y} \belief \geq 0$.
This implies that $ c_a ^\p B_{y} \belief$  satisfies a single crossing condition in $(y,a)$.
Therefore $a_n(\belief,y) = \argmin_a c_a ^\p B_{y} \belief$ is increasing in $y$ for any belief $\belief$.

\subsection*{Proof of Theorem \ref{thm:socialincestfilter}}

The local estimate at node $n$ is given by (\ref{eq:socialconstraintestimate}), namely,
\beq \lbelief_{n-}(i) = w_n^\p  \lbelief_{1:n-1}(i).  \label{eq:inproof}
\eeq
Define $\bar{R}^\belief_{ia} = \log P(a|x=i,\belief)$ and  the $n-1$ dimensional vector $\bar{R}_{1:n-1}(i) =  \begin{bmatrix}  \bar{R}^{\belief_1}_{i,a_1} & & \bar{R}^{\belief_{n-1}}_{i,a_{n-1}}  \end{bmatrix} $.
From the structure of  transitive closure matrix $T_n$,
\beq  \lbelief_{1:n-1}(i)  = T_{n-1}^\p  \,\loprob_{1:n-1}(i),  \;  \lbelief_{n-}(i) = t_n^\p  \loprob_{1:n-1}(i) \label{eq:inproof2} \eeq
Substituting the first equation in (\ref{eq:inproof2}) into (\ref{eq:inproof})  yields $$  \lbelief_{n-}(i) = w_n^\p T_{n-1}^\p  \,\loprob_{1:n-1}(i). $$
Equating this with the second equation in (\ref{eq:inproof2}) yields $ w_n =  T_{n-1}^{-1}  t_n$.
(By Lemma \ref{lem:properties},
  $T_{n-1}$ is  invertible).

\subsection*{Proof of Theorem \ref{thm:estpollster}}
Given $n$ nodes, it is clear from Bayes formula and the structure of the adjacency matrix $A$ of the DAG that
$$ \lbelief_{1:n}(i) = \logoprob_{1:n} +  e_1 \lbelief_0(i) + A^\p  \lbelief_{1:n}(i) $$
Since $I-A$ is  invertible by construction,
$$  \lbelief_{1:n}(i) = (I - A^\p)^{-1}  \logoprob_{1:n} + (I - A^\p)^{-1}  e_1  \lbelief_0(i) $$
Then   $\lbelief_{\Req}(i) =  \begin{bmatrix} 
e _{\Req_1} &  e_{\Req_2} & \cdots &  e _{\Req_L}  \end{bmatrix}^\p \lbelief_{1:n}(i)$ satisfies (\ref{eq:omat}).
Finally $P(x|\belief_\Req) \propto \sum_{Y_{1:n} \in \obseq}  P(Y_{1:n}|x) \,  \belief_0(x)$.
Here $\obseq$ is the set of $n$ dim.\ vectors satisfying~(\ref{eq:omat}).



\vspace{-1.2cm}

\begin{IEEEbiographynophoto}{Vikram Krishnamurthy}
[F]  (vikramk@ece.ubc.ca)   is  a professor and Canada Research Chair at the
Department of Electrical Engineering, University of British Columbia,
Vancouver, Canada. Dr Krishnamurthy's  current research interests include statistical signal processing, computational game theory and
stochastic control in social networks. He served as distinguished lecturer for the IEEE Signal Processing Society and
Editor in Chief of IEEE Journal Selected Topics in Signal Processing. He received an honorary doctorate from KTH (Royal Institute of Technology), Sweden
in 2013.
\end{IEEEbiographynophoto}

\begin{IEEEbiographynophoto}{William Hoiles}
 (whoiles@ece.ubc.ca) is currently a Ph.D. candidate in the Department of Electrical and Computer Engineering, University of British Columbia, Vancouver, Canada. He received the M.A.Sc. degree in 2012 from the Department of Engineering Science, Simon Fraser University, Vancouver, Canada. His current research interests are social sensors and the bioelectronic interface.
\end{IEEEbiographynophoto}

\vfill
\end{document}